\documentclass[11pt]{article}

\usepackage{amsmath, amsthm, amssymb, amscd}
\usepackage{graphicx}
\usepackage{ifthen,color}
\usepackage[title]{appendix}

\usepackage{hyperref}

\newcommand{\newoddside}{
        \ifthenelse{\iseven{\thepage}}{\newpage}{
        \newpage
        \textcolor{white}{placeholder} 
        \thispagestyle{empty}
        \newpage
        }
}
\newcommand{\newevenside}{
        \ifthenelse{\isodd{\thepage}}{\newpage}{
        \newpage
        \textcolor{white}{placeholder} 
        \thispagestyle{empty}
        \newpage
        }
}

\numberwithin{equation}{section}
\newtheorem{thm}{Theorem}
\newtheorem{lem}{Lemma}
\newtheorem{rem}{Remark}

\newtheorem{defi}{Definition}

\newtheorem{thm2}{Theorem}[section]
\newtheorem{lem2}{Lemma}[section]
\newtheorem{rem2}{Remark}[section]
\newtheorem{prop2}{Proposition}[section]
\newtheorem{cor2}{Corollary}[section]

\newtheorem{defi2}{Definition}[section]

\newcommand{\Var}{\operatorname{Var}}
\newcommand{\Cov}{\operatorname{Cov}}

\newcommand{\E}{\operatorname{E}}

\newcommand{\Bias}{\operatorname{Bias}}

\newcommand{\MISE}{\operatorname{MISE}}

\newcommand{\Dir}{\operatorname{Dir}}
\newcommand{\MA}{\operatorname{MA}}
\newcommand{\Tr}{\operatorname{tr}}
\newcommand{\Cum}{\operatorname{Cum}}
\newcommand{\TrSq}{\operatorname{TrSq}}

\def\liminf{\mathop{\underline{\rm lim}}}

\begin{document}

\title{Nonparametric Estimation of the Volatility Function in a High-Frequency Model corrupted by Noise}

\author{Axel Munk$^*$ and Johannes Schmidt-Hieber \\
\vspace{0.1cm} \\
{\em Institut f\"ur Mathematische Stochastik, Universit\"at G\"ottingen,} \\ {\em Goldschmidtstr. 7, 37077 G\"ottingen} \\ {\small {\em Email:} \texttt{munk@math.uni-goettingen.de}, \texttt{schmidth@math.uni-goettingen.de}}}


\date{}
\maketitle

\begin{abstract}
We consider the models $Y_{i,n}=\int_0^{i/n} \sigma(s)dW_s+\tau(i/n)\epsilon_{i,n}$, and $\tilde  Y_{i,n}=\sigma(i/n)W_{i/n}+\tau(i/n)\epsilon_{i,n}$, \ $i=1,\ldots,n$, where $\left(W_t\right)_{t\in \left[0,1\right]}$ denotes a standard Brownian motion and $\epsilon_{i,n}$ are centered i.i.d. random variables with $\E\left(\epsilon_{i,n}^2\right)=1$ and finite fourth moment. Furthermore, $\sigma$ and $\tau$ are unknown deterministic functions and $\left(W_t\right)_{t\in \left[0,1\right]}$ and $\left(\epsilon_{1,n},\ldots,\epsilon_{n,n}\right)$ are assumed to be independent processes. Based on a spectral decomposition of the covariance structures we derive series estimators for $\sigma^2$ and $\tau^2$ and investigate their rate of convergence of the $\MISE$ in dependence of their smoothness. To this end specific basis functions and their corresponding Sobolev ellipsoids are introduced and we show that our estimators are optimal in minimax sense. Our work is motivated by microstructure noise models. A major finding is that the microstructure noise $\epsilon_{i,n}$ introduces an additionally degree of ill-posedness of $1/2$; irrespectively of the tail behavior of $\epsilon_{i,n}$. The performance of the estimates is illustrated by a small numerical study.
\end{abstract}

\medskip



\section{Introduction}
Consider the models
\begin{eqnarray}
	Y_{i,n}= \int_0^{i/n}\sigma\left(s\right)dW_s
	+\tau\left(\frac in\right) \epsilon_{i,n} 
	\quad i=1,\ldots,n,
	\label{eq.mod2}
\end{eqnarray}
and
\begin{eqnarray}
	\tilde Y_{i,n}= \sigma\left(\frac in\right)W_{i/n}
	+\tau\left(\frac in\right) \epsilon_{i,n} 
	\quad i=1,\ldots,n
	\label{eq.mod}
\end{eqnarray}
respectively, where $\left(W_t\right)_{t\in \left[0,1\right]}$ denotes a Brownian motion and $\epsilon_{i,n}$ is so called microstructure noise, i.e. we assume $\epsilon_{i,n}$ i.i.d., $\E\left(\epsilon_{i,n}^2\right)=1$ and $\E\left(\epsilon_{i,n}^4\right)<\infty$. $\left(W_t\right)_{t\in \left[0,1\right]}$ and $\left(\epsilon_{1,n},\ldots,\epsilon_{n,n}\right)$ are assumed to be independent, and $\sigma$ and $\tau$ are unknown, positive and deterministic functions. 

Our models (\ref{eq.mod2}) and (\ref{eq.mod}) are natural extensions of the situation when $\sigma$ and $\tau$ are constant, which has been, in a slightly broader setting, previously considered by \cite{cai}, \cite{glo1}, \cite{glo2} and \cite{ste} among others. In the latter papers sharp minimax estimators were derived for $\sigma^2$ and $\tau^2$. The minimax rate for $\sigma^2$ is $n^{-1/4}$ and for $\tau^2$ it is $n^{-1/2}$, and the corresponding constants for quadratic loss (MSE) being $8 \tau \sigma^3$ and $2 \tau^4$, respectively. To estimate $\sigma$ and $\tau,$ maximum likelihood is feasible (see \cite{ste}) and achieves these bounds.
Other efficient estimators where given by \cite{cai}, \cite{glo1} or \cite{glo2}.
In our case, i.e. when $\sigma$ and $\tau$ are functions these methods fail
and techniques from nonparametric regression become necessary. We will postpone a more careful dicussion of models (\ref{eq.mod2}) and (\ref{eq.mod}) to Section \ref{sec.dmn}.

Both models incorporate, as usually in high-frequency financial models, an additional noise term, denoted 
as microstructure noise (cf. \cite{ban} and \cite{mad} ) in order to model market frictions such as bid-ask spreads and rounding errors. In general, microstructure noise is often assumed as white noise process with bounded fourth moment. Therefore, we may interpret both models as obtaining data from transformed Brownian motions under additional measurement errors. Particularly, our assumptions cover the important case when $\epsilon_{i,n}\stackrel{i.i.d.}{\sim}\mathcal{N}\left(0,1\right).$ 

In this paper we try to understand how estimation of the functions $\sigma^2$ and $\tau^2$ in (\ref{eq.mod2}) and (\ref{eq.mod}) itself can be performed, i.e. the time derivative of the integrated volatility. To our knowledge, this issue has never been addressed before, a remarkable exception is \cite{baru} where a harmonic  
analysis technique is introduced in order to recover $\sigma^2$ from noiseless data. A naive estimator of $\sigma^2$ would be the derivative of an estimator of $\int_0^s \sigma^2(x)dx$ with respect to $s$. However, (numerical) differentiation of  $\int_0^s\sigma^2(x)dx$ with respect to $s$  yields an additional degree of ill-posedness and there are to the best of our knowledge no estimates and no theoretical results available how to estimate $\sigma^2$ in our situation. Instead, we propose a regularized estimator for $\sigma$ and $\tau$ that attains the minimax rate of convergence. Our estimator is a Fourier series estimator where we will estimate the single cosine Fourier coefficients, $\int_0^1 \sigma^2(x)\cos(k\pi x)dx$, $k=0,1,\ldots$ by a particular spectral estimator which is specifically tailor suited to this problem. The difficulty to estimate $\sigma^2$  can be explained generically from the point of view of statistical inverse problem: Microstructure noise induces an additional degree of ill posedness -similar as in a deconvolution problem- which in our case leads to a reduction of the rate of convergence by a factor $1/2$. Surprisingly, and in contrast to deconvolution, this is only reflected in the behavior of the eigenvalues of the covariance operator of the process in (\ref{eq.mod2}) and (\ref{eq.mod}) and not in the tail behavior of the Fourier transform of the error $\epsilon_{i,n}$.

We stress again that we are aware of the fact that our model assumes a deterministic function $\sigma$ and $\tau$, which only depends on time $t$ and generalization to $\sigma\left(t,X_t\right)$ is not obvious and a challenge for further research. However, the purely deterministic case already helps us to reveal the daily pattern of the volatility and finally we believe that our analysis is an important step into the understanding of these models from the view point of a statistical inverse problems.

\noindent
{\em Results:} All results are obtained with respect to $\MISE$-risk. Let $\alpha$ and $\beta$ denote a certain smoothness of $\sigma^2$ and $\tau^2$, respectively. Roughly speaking, these numbers correspond to the usual Sobolev indices, although in our situation, a particular choice of basis is required, leading us to the definition of Sobolev {\it s}-ellipsoids (see Definition \ref{def.ssobo}). Then we show that $\tau^2$ can be estimated at rate $n^{-\beta/\left(2\beta+1\right)}$ for $\beta>1, \alpha>1/2$ in model (\ref{eq.mod2}) and $\beta>1, \alpha>3/4$ in model (\ref{eq.mod}). This corresponds to the classical minimax rates for the usual Sobolev ellipsoids without the Brownian motion term in (\ref{eq.mod2}) and (\ref{eq.mod}). More interesting,  we obtain for estimation of $\sigma^2$ the $n^{-\alpha/\left(4\alpha+2\right)}$ rate of convergence for $\alpha>3/4, \beta>5/4$ in model (\ref{eq.mod2}) and $\alpha>3/2, \beta>5/4$ in model (\ref{eq.mod}). 
We will show that these rates are uniform for Sobolev {\it s}-ellipsoids. Lower bounds with respect to H\"older classes for estimation of $\sigma^2$ have been obtained in \cite{mun}. Here we will extend this result to Sobolev {\it s}-ellipsoids. It follows that the obtained rates are minimax, indeed.

To summarize, our major finding is that in contrast to ordinary deconvolution the difficulty of estimation $\sigma^2$ when corrupted by additional (microstructure) noise $\epsilon$, is generically increased by a factor of $1/2$ within the {\it s}-ellipsoids. This is quite surprising because one might have expected that for instance Gaussian error leads to logarithmic convergence rates due to its exponential decay of the Fourier transform (see e.g. \cite{BHMR07}, \cite{BT07a}, \cite{BT07b} and \cite{Fan91} for some results in this direction). We stress that for our method a minimal smoothness of $\sigma$ in (\ref{eq.mod2}) of $\alpha>1/2$ and in (\ref{eq.mod}) of $\alpha>3/2$ is required. Although convergence rates are half compared with usual nonparametric regression, it turns out that for large sample sizes we get reasonable estimates for smooth functions $\sigma^2$. Roughly speaking, the results imply that $n$ data points for estimation of $\sigma^2$ can be compared to the situation, when we have $\sqrt{n}$ observation in usual nonparamteric regression.

The work is organized as follows. In Sections \ref{sec.dmn} and \ref{sec.techpre} we will discuss models (\ref{eq.mod2}) and (\ref{eq.mod}) in more detail, introduce notation and define the required smoothness classes, Sobolev {\it s}-ellipsoids (details can be found in Appendix \ref{sec.ass}). Section \ref{sec.esttau} and Section \ref{sec.estsigma} are devoted to estimate $\sigma^2$ and $\tau^2$, respectively, and to present the rates of convergence of the estimators (for a proof see Appendix \ref{sec.cvr}). Section \ref{sec.minimax} provides the minimax result. In Section \ref{sec.sim} we briefly discuss some numerical results and illustrate the robustness of the estimator against non-normality and violations of the required smoothness assumptions for $\sigma^2$ and $\tau^2$. Some further results and technicalities of Sections \ref{sec.esttau} and \ref{sec.estsigma} are given in the supplementary material.

\section{Discussion of Models (\ref{eq.mod2}) and (\ref{eq.mod})}
\label{sec.dmn}

In this subsection we briefly discuss the background from financial economics of model (\ref{eq.mod2}) and explore the differences between models (\ref{eq.mod2}) and (\ref{eq.mod}). We may consider the processes $\left(\sigma\left(t\right)W_t\right)_{t\in \left[0,1\right]}$ and $\left(\int_0^t\sigma\left(s\right)dW_s\right)_{t\in \left[0,1\right]}\stackrel{\mathcal{D}}{=}\left(W\left(H\left(t\right)\right)\right)_{t\in \left[0,1\right]}$, $H\left(t\right):=\int_0^t\sigma^2\left(s\right)ds$ as (inhomogeneously) scaled Brownian motions, where scaling takes place in space and in time, respectively. Hence we will refer to $\left(\sigma\left(t\right)W_t\right)_{t\in \left[0,1\right]}$ and $\left(\int_0^t\sigma\left(s\right)dW_s\right)_{t\in \left[0,1\right]}$ in the future as space-transformed (sBM) and time-transformed (tBM) Brownian motion.

\bigskip

{\bf Model (\ref{eq.mod2}):} In the financial econometrics literature variations of model (\ref{eq.mod2}) are often denoted as high-frequency models, since $\left(W_t\right)_{t\in \left[0,1\right]}$ is sampled on time points $t=i/n$ and nowadays there is a vast amount of literature on volatility estimation in high-frequency models with additional microstructure noise term (see \cite{barn}, \cite{jac}, \cite{zha2} and \cite{zha1}).  These kinds of models have attained a lot of attention recently, since the usual quadratic variation techniques for estimation of $\int_0^1\sigma^2(x)dx$ lead to inconsistent estimators (cf. \cite{zha2}).

We are aware of the fact, that in contrast to our model,  volatility is modelled generally not only as time dependent but also depending on the process itself, i.e. $Y_{i,n}=X_{i/n}+\tau\left(i/n\right)\epsilon_{i,n}$, $i=1,\ldots,n,$ $dX_t=\sigma\left(t,X_t\right)dW_t.$ An overview over commonly used parametric forms of $\sigma\left(t,X_t\right)$ and a non-parametric treatment in the absence of microstructure noise, can be found in \cite{fan}.
It is known that the same rates as for the case $\sigma$ and $\tau$ constant hold true if we consider the model (\ref{eq.mod2}) and estimate the so called integrated volatility or realized volatility $\int_0^s\sigma^2(x)dx$  ($s\in[0,1]$) and $\int_0^s\tau^2(x)dx$ instead of $\sigma^2$ and $\tau^2$, 
respectively (see \cite{pod} and \cite{ros} for a discussion on estimation of integrated volatility and related quantities). Recently, model (\ref{eq.mod2}) has been proven to be asymptotically equivalent to a Gaussian shift experiment (see \cite{rei2}). $\sigma^2$ as a function of time corresponds in model (\ref{eq.mod2}) to the instantaneous volatility or spot volatility.

\bigskip

{\bf Model (\ref{eq.mod}):} Model (\ref{eq.mod}) can be regarded as a nonparametric extension of the model with constant $\sigma, \tau$ as discussed for variogram estimation by \cite{ste}. To motivate the usefulness of sBM we give the following Lemma.

\begin{lem}
\begin{itemize}
\item[(i)] Assume that $\sigma$, $0<c\leq\sigma,$ is continuously differentiable. Then the corresponding sBM, $\left(\sigma\left(t\right)W_t\right)_{t\in \left[0,1\right]}$ is the unique solution of the SDE
\begin{align*}
	dX_t= X_t \ d\left(\log\left(\sigma\left(t\right)\right)\right)
	+\sigma\left(t\right)dW_t, \quad X_0=0, \ 0\leq t\leq T.
\end{align*}
\item[(ii)] The variogram of sBM is given by
\begin{align*}
	\gamma\left(s,t\right)&:=
	\E\left(X_t-X_s\right)^2 \\
	&=
	\left(\sigma\left(t\right)t^{1/2}-\sigma\left(s\right)s^{1/2}\right)^2
	+\sigma\left(t\right)\sigma\left(s\right)
	\left[
	\left|s-t\right|-\left(s^{1/2}-t^{1/2}\right)^2
	\right].
\end{align*}
\end{itemize}
\end{lem}

\begin{proof}
{\it (i)} It is easy to check that sBM indeed is a solution. To establish uniqueness, we apply Theorem 9.1 in \cite{steele}.
{\it (ii)} This follows by straightforward calculations.
\end{proof}

\bigskip

{\bf Comparison of the models:} We remark that tBM  can be related to sBM by partial integration $\int_0^t \sigma\left(s\right)dW_s = \sigma\left(t\right)W_t-\int_0^t \sigma'\left(s\right)W_s ds.$ To see the differences we compared in Figure \ref{fig:pdiff} sBM and tBM in two typical situations: The case where $\sigma\left(t\right)=0$ for $t>T$ and the case, where $\sigma$ is non-continuous. If $\sigma\left(t\right)=0$ for $t>T,$ sBM tends to zero, whereas tBM tends to a constant, i.e. the random variable $\int_0^{T}\sigma\left(s\right)dW_s$. Furthermore, if $\sigma$ is a jump function, sBM has a jump too, whereas tBM does not.

Unlike Model (\ref{eq.mod2}), which can be viewed as a price process, Model (\ref{eq.mod}) has no direct application in financial mathematics. However, from the view point of nonparametric statistics it seems to be a natural extension of the situation when $\sigma$ and $\tau$ are constant.

\section{Introduction to Sobolev s-Ellipsoids and Technical Preliminaries}
\label{sec.techpre}
In this section we shortly introduce the setup needed in order to define the estimators. First we define suitable smoothness classes, which are different, but related to well known Sobolev ellipsoids (see Definition \ref{def.sobell}).
\begin{defi}
\label{def.ssobo}
For $\alpha>0,$ $C>0$, we call the function space
\begin{align*}
	&\Theta_s:=\Theta_s(\alpha,C)
	:= \\
	& \ \left\{f \in L^2[0,1]: \exists \left(\theta_n\right)_{n\in \mathbb{N}}, \ \text{s. t.} \ f(x)=\theta_0+2\sum_{i=1}^\infty \theta_i\cos\left(i\pi x\right),\sum_{i=1}^\infty i^{2\alpha}\theta_i^2 \leq C\right\}
\end{align*}
a Sobolev {\it s}-ellipsoid. If there is a $C<\infty$ such that $f\in \Theta_s\left(\alpha,C\right)$, we say $f$ has smoothness $\alpha$. For $0<l< u <\infty$, we further introduce the uniformly bounded Sobolev {\it s}-ellipsoid
\begin{align*}
	\Theta^b_s(\alpha,C):=\Theta^b_s(\alpha,C, [l, u])
	:= \left\{f \in \Theta_s(\alpha,C): l\leq f\leq u \right\}.
\end{align*}
\end{defi}
Here the ``{\it s}'' refers to ``{\it symmetry}'' since the $L^2[0,1]$ basis
\begin{eqnarray}
	\left\{\psi_k, k=0,\ldots \right\}:=\left\{1, \sqrt{2}\cos\left(k\pi t\right), k=1,\ldots\right\},
	\label{eq.cosbas}
\end{eqnarray}
can also be viewed as a basis of the symmetric $L^2[-1,1]$ functions $$\left\{f : f\in L^2\left[-1,1\right], f(x)=f(-x) \ \forall x\in [0,1]\right\}.$$ Usually, Sobolev ellipsoids are introduced with respect to the Fourier basis $$\left\{1, \sqrt{2}\sin\left(2k\pi t\right), \sqrt{2} \cos\left(2k\pi t\right), k=1,\ldots\right\}$$ on $L_2\left[0,1\right]$ (see Definition (\ref{def.sobell})). As will turn out later on, Sobolev {\it s}-ellipsoids are more natural for our approach. If a function has a certain smoothness in one space, it might have a completely different smoothness with respect to the other basis. For instance the function $\cos\left(\left(2l+1\right)\pi x\right)$, $l\in \mathbb{N}$ has smoothness $\alpha$ for all $\alpha<\infty$ with respect to basis (\ref{eq.cosbas}), and as can be seen by direct calculations only smoothness $\alpha<1/2$ for the Fourier basis. A more precise discussion can be found in Part \ref{sec.ass} of the Appendix.

Instead of (\ref{eq.cosbas}) it is convenient to introduce the functions $f_k:\left[0,1\right]\rightarrow \mathbb{R}$, $k \in \mathbb{N}$
\begin{eqnarray*}
	f_k(x):= \psi_k\left(\frac x2\right).
\end{eqnarray*}
Note that for $k\geq 1$, $f_k^2$ can be expanded in basis (\ref{eq.cosbas}) by  $f_k^2=\psi_0+2^{-1/2}\psi_k.$ For any function $g$ we introduce the forward difference operator $\Delta_i g := g\left(\left(i+1\right)/n\right)-g\left(i/n\right)$ and further the transformed variables $\Delta Y^{k,1}_{i,n}:=\left(Y_{i+1,n}- Y_{i,n}\right)f_k\left(i/n\right)$ and $\Delta Y^{k,2}_{i,n}:=\left(\tilde Y_{i+1,n}- \tilde Y_{i,n}\right)f_k\left(i/n\right)$, $i=1,\ldots,n-1$ for models (\ref{eq.mod2}) and (\ref{eq.mod}), respectively. In order to discuss the models simultaneously, we will write $\Delta Y^{k}_{i,n}=\Delta Y^{k,l}_{i,n}$, $l=1,2.$ Throughout the paper we abbreviate first order differences of observations by
\begin{eqnarray*}
	\Delta Y^k
	:=
	\left(
	\Delta Y^k_{1,n},\ldots,\Delta Y^k_{n-1,n}\right)^t
	.
\end{eqnarray*}

We write $\mathbb{M}_{p,q}$, $\mathbb{M}_p$ and $\mathbb{D}_p$ for the space of $p \times q$ matrices, $p \times p$ matrices and $p \times p$ diagonal matrices over $\mathbb{R}$, respectively. Further let $D_{n-1} \in \mathbb{M}_{n-1}$ given by $\left(D_{n-1}\right)_{i,j}=\sqrt{2/n}\sin\left(ij\pi/n\right)$ and define
\begin{eqnarray}
	\lambda_{i,n-1}:=4\sin^2\left(i\pi/\left(2n\right)\right)
	\quad i=1,\ldots,n-1 \ ,
	\label{def.lambdai}
\end{eqnarray}
the eigenvalues of the covariance matrix $K_{n-1} \in \mathbb{M}_{n-1}$ of the $\MA(1)$ process $\Delta_i\epsilon_{i,n}:=\epsilon_{i+1,n}-\epsilon_{i,n}$, $i=1,\ldots,n-1$.   More explicitly $K_{n-1}$ is tridiagonal and
\begin{eqnarray}
	\left(K_{n-1}\right)_{i,j}
	=
	\begin{cases}
	2 \quad \text{for} \quad i=j \\
	-1 \quad \text{for} \quad \left|i-j\right|=1 \\
	0 \quad \text{else}
	\end{cases}
	.
	\label{def.k}
\end{eqnarray}
Note that we can diagonalize $K_{n-1}$ explicitly by $K_{n-1}=D_{n-1}\Lambda_{n-1} D_{n-1}$, where $\Lambda_{n-1}$ is diagonal with diagonal entries given by (\ref{def.lambdai}).

We will suppress the index $n-1$ and write $K$, $D$, $\Lambda$, $\lambda_i$ instead of $K_{n-1}$, $D_{n-1}$, $\Lambda_{n-1}$, and $\lambda_{i,n}$, respectively. We write $\left[x\right]:=\max_{z \in \mathbb{Z}}\left\{z\leq x\right\}$, $x\in \mathbb{R}$, the integer part of $x$. $\log()$ is defined to be the binary logarithm and in order to define estimators properly, we assume throughout the paper additionally $n> 16.$

\section{Estimators and Rates of Convergence}

\subsection{Estimation of $\tau^2$}
\label{sec.esttau}

Before we will turn to the estimation of the volatility $\sigma^2$, we will first discuss estimation of the noise variance, i.e. $\tau^2$. Let $J^\tau_n \in \mathbb{D}_{n-1}$ given by
\begin{eqnarray*}
	\left(J^{\tau}_n\right)_{i,j}:=
	\begin{cases}
	\left(n-n/\log n\right)^{-1}
	\lambda_i^{-1}\delta_{i,j}, \quad \text{for} \quad \left[n/\log n\right]\leq i,j\leq n-1 \\
	0 \quad \text{otherwise}
	\end{cases}
	,	
\end{eqnarray*}
where $\lambda_i$ is defined by (\ref{def.lambdai}) and $\delta_{i,j}$ denotes the Kronecker delta. We consider models (\ref{eq.mod2}) and (\ref{eq.mod}), simultaneously. Let
\begin{eqnarray}
	\hat t_{k,0}
	:=
	\left(\Delta Y^k\right)^t D J^{\tau}_n D^t \left(\Delta Y^k\right).
	\label{eq.deftk}
\end{eqnarray}
In Lemma \ref{lem.tlem} it will be shown that $\hat t_{k,0}$ is a $\sqrt{n}-$consistent estimator of $$t_{k,0}:= \newline \int_0^1 \tau^2(x)f_k^2\left(x\right)dx.$$ Note that for $k\geq1$ this means $t_{k,0}= \int_0^1 \tau^2(x)\psi_0(x)dx+2^{-1/2} \int_0^1 \tau^2(x)\psi_k(x)dx$. Define $Z:=D\left(\Delta Y^k\right)$ and denote by $Z_i$ the $i$-th component of $Z$. Then
\begin{eqnarray}
 	\hat t_{k,0}= \left(n-n/\log n\right)^{-1}
	\sum_{i=\left[n/\log n\right]}^{n-1} \lambda_i^{-1}Z_i^2.
	\label{eq.alttau}
\end{eqnarray}
Hence this also can be seen as a spectral filter in Fourier domain, where we cut off the first $n/\log n \ $ frequencies. Note that for $i\geq 1$, $2^{1/2}\left(t_{i,0}-t_{0,0}\right)=\int_0^1\tau^2(x)\psi_i(x)dx$ is the $i$-th series coefficient with respect to basis (\ref{eq.cosbas}). This observation suggests to construct the cosine series estimator
\begin{eqnarray}
	\hat \tau^2_N(t) := \hat t_{0,0} +2\sum_{i=1}^N \left(\hat t_{i,0}
	-\hat t_{0,0} \right)\cos\left(i\pi t\right).
	\label{eq.tsdef} 
\end{eqnarray}

The next result provides the rate of convergence of $\hat \tau^2_N$ uniformly within Sobolev s-ellipsoids. To this end a version of the continuous Sobolev embedding theorem is required for non-integer indices $\alpha, \beta$ (see Lemma \ref{lem.sobolem}). A proof of the following Theorem can be found in the supplementary material.
\begin{thm}[$\MISE$ of $\hat \tau^2_N(t)$]
\label{thm.misetthm}
Let $\hat \tau_N^2(t)$ as defined in (\ref{eq.tsdef}). Assume $\beta>1,$ and $Q, \bar Q>0$. Further suppose that $N=N_n = o\left(n^{1/2}/\log n\right).$ Assume either model (\ref{eq.mod2}) and $\alpha>1/2$ or model (\ref{eq.mod}) and $\alpha>3/4$. Then it holds 
\begin{align*}
	\sup_{\sigma^2\in \Theta^b_s\left(\alpha,Q\right), \tau^2\in \Theta^b_s\left(\beta,\bar Q\right)}
	\MISE\left(\hat \tau_N^2\right)
	=O\left(N^{-2\beta}+Nn^{-1}\right).
\end{align*}
Minimizing the r.h.s. yields $N^*=O\left(n^{1/\left(2\beta+1\right)}\right)$ and consequently
\begin{align*}
	\sup_{\sigma^2\in \Theta^b_s\left(\alpha,Q\right), \tau^2\in \Theta^b_s\left(\beta,\bar Q\right)}
	\MISE\left(\hat \tau_{N^*}^2\right)=O\left(n^{-2\beta/\left(2\beta+1\right)}\right).
\end{align*}
\end{thm}

\begin{rem}
Note that for model (\ref{eq.mod2}) Theorem \ref{thm.misetthm} holds, whenever $\alpha>1/2$. Hence the Brownian motion part of the model can be viewed as a nuisance parameter, not affecting rates for estimation of $\tau^2$. However, for model (\ref{eq.mod}) $\alpha>3/4$ is required here. This more restrictive assumption is essentially a consequence of the fact that the process $\sigma\left(i/n\right)W_{i/n}$ is in general no
martingale.
\end{rem}

\begin{rem}
The result from Theorem \ref{thm.misetthm} can be extended to $1/2<\beta\leq 1$ in model (\ref{eq.mod2}) and to $1/2<\alpha\leq 3/4$, $1/2<\beta\leq 1$ in model (\ref{eq.mod}). 
Let $\tilde t_{k,0}$ be defined as $\hat t_{k,0}$ in (\ref{eq.deftk}) but $J_n^\tau$ is now replaced by $\tilde J_n^\tau \in \mathbb{D}_{n-1},$
\begin{align*}
	\left(\tilde J_n^\tau\right)_{i,j}
	=
	\begin{cases}
	2n^{-1}
	\lambda_i^{-1}\delta_{i,j}, \quad \text{for} \quad \left[n/2\right]\leq i,j\leq n-1 \\
	0 \quad \text{otherwise}
	\end{cases}
	.
\end{align*}
Introduce further the estimator $\tilde \tau^2_N\left(t\right)= \tilde t_{0,0} +2\sum_{i=1}^N \left(\tilde t_{i,0} -\tilde t_{0,0} \right)\cos\left(i\pi t\right).$ Further suppose that $N =O\left(n^{1/\left(2\beta+1\right)}\right).$ Then we obtain by slight modifications of the proof of Theorem \ref{thm.misetthm} for $\beta>1/2,$  $\alpha>1/2$ and $Q, \bar Q >0$ 

\begin{itemize}
\item[(i)] Assume model (\ref{eq.mod2}). Then it holds
\begin{eqnarray*}
	\sup_{\sigma^2\in \Theta^b_s\left(\alpha,Q\right), \tau^2\in \Theta^b_s\left(\beta,\bar Q\right)}
	\MISE\left(\tilde \tau_N^2\right)
	=O\left(N^{-2\beta}+Nn^{-1}+Nn^{1-2\beta}\right)
\end{eqnarray*}
and $N^* = O\left(n^{\left(2\beta-1\right)/\left(2\beta+1\right)}\right)$ yields  
\begin{align*}
	\sup_{\sigma^2\in \Theta^b_s\left(\alpha,Q\right), \tau^2\in \Theta^b_s\left(\beta,\bar Q\right)}
	\MISE\left(\tilde \tau_{N^*}^2\right)=O\left(n^{-\left(4\beta^2-2\beta\right)/\left(2\beta+1\right)}\right).
\end{align*}
\item[(ii)] Assume model (\ref{eq.mod}). Then we have the expansion
\begin{align*}
	\sup_{\sigma^2\in \Theta^b_s\left(\alpha,Q\right), \tau^2\in \Theta^b_s\left(\beta,\bar Q\right)}
	&
	\MISE\left(\tilde \tau_N^2\right)
	=
	O\left(N^{-2\beta}+Nn^{-1}+Nn^{1-2\beta}+Nn^{2-4\alpha}\right),
\end{align*}
and the choice
\begin{eqnarray*}
	N^* =
	\begin{cases}
	O\left(n^{\left(2\beta-1\right)/\left(2\beta+1\right)}\right) \quad \text{for}
	\ \ \beta\leq1\wedge\left(2\alpha-1/2\right), \\
	O\left(n^{\left(4\alpha-2\right)/\left(2\beta+1\right)}\right)
	\quad \text{for}
	\ \ \alpha\leq3/4\wedge\left(\beta/2+1/4\right)
	\end{cases}
\end{eqnarray*}
yields
\begin{align*}
	&\sup_{\sigma^2\in \Theta^b_s\left(\alpha,Q\right), \tau^2\in \Theta^b_s\left(\beta,\bar Q\right)}
	\MISE\left(\tilde \tau_{N^*}^2\right) \\
	&\quad \quad \quad 
	=
	\begin{cases}
	O\left(n^{-\left(4\beta^2-2\beta\right)/\left(2\beta+1\right)}\right)  \quad &\text{for}
	\ \ \beta\leq1\wedge\left(2\alpha-1/2\right), \\
	O\left(n^{-\left(2-2\alpha\right)/\left(2\beta+1\right)}\right) \quad &\text{for}
	\ \ \alpha\leq3/4\wedge\left(\beta/2+1/4\right).
	\end{cases}
\end{align*}
\end{itemize}
\end{rem}

\begin{rem}
It is also possible, although more technical, to compute the asymptotic constant of the estimator $\hat \tau^2_{N^*}$. Suppose that the microstructure noise is Gaussian and assume model (\ref{eq.mod2}) and $\beta>1$ or (\ref{eq.mod}) and $\beta>1, \alpha>3/4$, then we have more explicitly
\begin{eqnarray*}
	\MISE\left(\hat \tau^2_{N^*}\right)
	= \frac{2N^*}n\int_0^1 \tau^4(x) dx + \sum_{k=N^*+1}^\infty \left(\int_0^1 \tau^2(x) \psi_k(x)dx\right)^2
	+o\left(N^*n^{-1}\right).
\end{eqnarray*}
\end{rem}

\begin{rem}
There are of course simpler estimators for $t_{k,0}$. For instance if we replace $J_n^\tau$ in (\ref{eq.deftk}) by $\left(2n\right)^{-1}\mathbb{I}_{n-1}$, where $\mathbb{I}_{n-1}\in \mathbb{D}_{n-1}$ denotes the identity matrix, we obtain the quadratic variation estimator for $t_{k,0}$ (cf. \cite{ban}) and it is not difficult to show that this estimator attains the optimal rate of convergence. This approach could even be extended to a nonparametric estimator of the form (\ref{eq.tsdef}). However, the single Fourier coefficients are not estimated efficiently, since in the case when the microstructure noise is Gaussian the asymptotic constant is $3n^{-1}\int \tau_k^4(x)dx$ (this is a straightforward extension of Theorem A.1 in \cite{zha1}) whereas for our estimator we have $2n^{-1}\int \tau_k^4(x)dx$ (see Lemma \ref{lem.tlem}). If $\tau$ is constant it can be easily seen that estimators in (\ref{eq.deftk}) are efficient for $k=0$ whereas quadratic variation is not.
\end{rem}

\begin{rem}
 In practical application it would be more natural to use instead of $n/\log n \ $ in (\ref{eq.alttau}) other cut-off frequencies e.g. $n^\gamma/\log n$ or $qn$, where $1/2< \gamma\leq 1$, $0<q<1$. Smaller $\gamma$ decreases the variance while on the other hand increases the bias of the estimator.
\end{rem}

\subsection{Estimation of $\sigma^2$}
\label{sec.estsigma}
Define $J_n \in \mathbb{D}_{n-1}$ by
\begin{eqnarray}
	\left(J_n\right)_{i,j}=
	\begin{cases}
	\sqrt{n}\delta_{i,j}, \quad \text{for} \quad \left[n^{1/2}\right]+1\leq i,j\leq 2\left[n^{1/2}\right] \\
	0 \quad \text{otherwise}
	\end{cases}
	.
	\label{eq.jndef}
\end{eqnarray}
Similar, as for the estimation of $\tau^2$ we first introduce an estimator of appropriate Fourier coefficients by
\begin{align}
	\hat s_{k,0}
	=
	\left(\Delta Y^k\right)^t D J_n D^t \left(\Delta Y^k\right)
	-7\pi^2\hat t_{k,0}/3.
	\label{eq.defsk}
\end{align}
The second part, i.e. $-7\pi^2\hat t_{k,0}/3$ is a bias correcting term, where the constant $7\pi^2/3$ is due to the choice of cut-off points $\left[n^{1/2}\right]+1$ and $2\left[n^{1/2}\right]$ in (\ref{eq.jndef}). As we will see, the estimator of $\hat t_{k,0}$ has better convergence properties than the first term in $\hat s_{k,0}$, and hence does not affect the asymptotic variance. Similar to  (\ref{eq.tsdef}), we put
\begin{eqnarray}
	\hat \sigma^2_N(t)= \hat s_{0,0}
	+2\sum_{i=1}^N \left(\hat s_{i,0}-\hat s_{0,0}\right)\cos\left(i\pi t\right).
	\label{eq.ssdef}
\end{eqnarray}

\begin{thm}[$\MISE$ of $\hat \sigma^2_N$]
\label{thm.misesthm}
Let $\hat \sigma_N^2$ as defined in (\ref{eq.ssdef}). Suppose that $N=N_n =o\left(n^{1/4}\right)$, $\beta>5/4$ and $Q, \bar Q>0$.
Assume model (\ref{eq.mod2}) and $\alpha>3/4$ or model (\ref{eq.mod}) and $\alpha>3/2$. Then it holds 
\begin{align*}
	\sup_{\sigma^2\in \Theta^b_s\left(\alpha,Q\right), \tau^2\in \Theta^b_s\left(\beta,\bar Q\right)}
	\MISE\left(\hat \sigma_N^2\right)
	= O\left(N^{-2\alpha}+Nn^{-1/2}\right)
\end{align*}
and minimizing the r.h.s. yields 
\begin{align*}
	\sup_{\sigma^2\in \Theta^b_s\left(\alpha,Q\right), \tau^2\in \Theta^b_s\left(\beta,\bar Q\right)}
	\MISE\left(\hat \sigma_{N^*}^2\right)
	= O\left(n^{-\alpha/\left(2\alpha+1\right)}\right)
\end{align*}
for $N^*=O\left(n^{1/\left(4\alpha+2\right)}\right)$.
\end{thm}

The proof of Theorem \ref{thm.misesthm} is given in Section \ref{subsec.estofsigma}.

\begin{rem}
It is also possible to extend this result for less smooth functions $\sigma^2$ and $\tau^2$.

\begin{itemize}
\item[(i)] Assume model (\ref{eq.mod2}) and $\alpha>1/2, \ \beta>1$. Then it holds
\begin{align*}
	&\sup_{\sigma^2\in \Theta^b_s\left(\alpha,Q\right), \tau^2\in \Theta^b_s\left(\beta,\bar Q\right)}
	\MISE\left(\hat \sigma_N^2\right) \\
	&\quad \quad = O\left(N^{-2\alpha}+Nn^{-1/2}
 	+Nn^{2-2\beta}+Nn^{1-2\alpha}\right),
\end{align*}
and
\begin{eqnarray*}
	N^* =
	\begin{cases}
	O\left(n^{\left(2\alpha-1\right)/\left(2\alpha+1\right)}\right) \quad \text{for}
	\ \ \alpha\leq 3/4\wedge \left(\beta-1/2\right), \\
	O\left(n^{\left(2\beta-2\right)/\left(2\alpha+1\right)}\right) \quad \text{for}
	\ \ \beta\leq 5/4\wedge \left(\alpha +1/2\right)
	\end{cases}
\end{eqnarray*}
yields
\begin{align*}
	&\sup_{\sigma^2\in \Theta^b_s\left(\alpha,Q\right), \tau^2\in \Theta^b_s\left(\beta,\bar Q\right)}
	\MISE\left(\hat \sigma_{N^*}^2\right) \\
	&\quad \quad =
	\begin{cases}
	O\left(n^{-2\alpha\left(2\alpha-1\right)/\left(2\alpha+1\right)}\right) \quad &\text{for}
	\ \ \alpha\leq 3/4\wedge \left(\beta-1/2\right) ,\\
	O\left(n^{-2\alpha\left(2\beta-2\right)/\left(2\alpha+1\right)}\right)  \quad &\text{for}
	\ \ \beta\leq 5/4\wedge \left(\alpha+1/2\right).
	\end{cases}
\end{align*}
\item[(ii)] Assume model (\ref{eq.mod}) and $\alpha>3/2, \ \beta>1$. Then it holds
\begin{eqnarray*}
	\sup_{\sigma^2\in \Theta^b_s\left(\alpha,Q\right), \tau^2\in \Theta^b_s\left(\beta,\bar Q\right)}
	\MISE\left(\hat \sigma_N^2\right)
	= O\left(N^{-2\alpha}+Nn^{-1/2}
	+Nn^{2-2\beta}\right),
\end{eqnarray*}
and $N^*= O\left(n^{\left(2\beta-2\right)/\left(2\alpha+1\right)}\right)$ yields
\begin{align*}
	\sup_{\sigma^2\in \Theta^b_s\left(\alpha,Q\right), \tau^2\in \Theta^b_s\left(\beta,\bar Q\right)}
	\MISE\left(\hat \sigma_{N^*}^2\right)= O\left(n^{-2\alpha\left(2\beta-2\right)/\left(2\alpha+1\right)}\right).
\end{align*}
\end{itemize}
\end{rem}

\begin{rem}
In analogy to (\ref{eq.alttau}), the estimator $\hat s_{k,0}$ can also be viewed as a spectral filter in Fourier domain, where essentially only the frequencies $n^{1/2},\ldots,2n^{1/2}$ play a role. For practical purposes one can generalize this to estimators where the frequencies $k,\ldots,\left[cn^{1/2}\right]$, $c>0$ are used. If $\sigma$ is assumed to be very smooth, one even may set $k=1$. In this more general setting, the constant $-7\pi^2/3$ in the definition of the estimator has to be replaced by $-n/\left(\left[cn^{1/2}\right]-k\right)\sum_{i=k}^{\left[cn^{1/2}\right]}\lambda_{i}.$
\end{rem}

\begin{rem}
\label{rem.timing}
Since the matrix $D$ in the definition of $\hat s_{k,0}$ is a discrete sine transform (for a definition see \cite{brit}) the estimator $\hat \sigma^2_N$ can be calculated explicitly taking $O\left(Nn\log n \right)$ steps.
\end{rem}

\section{Minimax}
\label{sec.minimax}

In this section we will discuss the optimality of the proposed estimators. To this end we establish lower bounds with respect to Sobolev {\it s}-ellipsoids.

\begin{thm}
Assume model (\ref{eq.mod2}) or model (\ref{eq.mod}), $\alpha\in\mathbb{N}\setminus \left\{0\right\}$. Further assume $\tau$ constant. Then there exists a $C>0$ (depending only on $\alpha, Q, l, u$), such that
\begin{eqnarray*}
	\liminf_{n \rightarrow \infty} \inf_{\hat \sigma^2_n} 
	\sup_{\sigma^2 \in \Theta^b_s\left(\alpha,Q\right)}
	\E\left(n^{\frac{\alpha}{2\alpha+1}}\left\|\hat \sigma^2-\sigma^2\right\|_2^2
	\right)\geq C.
\end{eqnarray*}
\end{thm}

\begin{proof}
The proof relies on a multiple hypothesis testing argument and is close to the proof given in \cite{mun}, Theorem 2.1. However, the lower bounds there are established with respect to the space of H\"older continuous functions
of index $\alpha$ on the interval $\left[0,1\right],$ i.e. for $l<u$
\begin{align*}
	&\mathcal{C}^b\left(\alpha, L\right)
	:= \mathcal{C}^b\left(\alpha, L, \left[l,u\right]\right) := 
	\left\{
	f: f^{\left(p\right)} \ \text{exists for } p=\left[\alpha\right],
	\right. \\  & \quad \left.
	\quad \left|f^{\left(p\right)}(x)-f^{\left(p\right)}(y)\right|
	\leq L\left|x-y\right|^{\alpha-p}, \ \forall x,y \in I, \ 0<l\leq f\leq u <\infty
	\right\}.
\end{align*}
Therefore, the statement above does not follow immediately from \cite{mun}, Theorem 2.1 because of $\mathcal{C}^b\left(\alpha,L\right)\varsubsetneq \Theta^b_s\left(\alpha,Q\right)$ due to boundary effects. Here, we will only point out the difference to the proof of \cite{mun}, Theorem 2.1. We write $\sigma_{\min}, \ \sigma_{\max}$ for the lower and upper bound of $\sigma^2$, respectively, i.e. $\sigma^2\in \Theta^b_s\left(\alpha,Q, \left[\sigma_{\min}, \sigma_{\max}\right] \right)$.  Without loss of generality, we may assume that $\sigma_{\min}=1$. For the multiple hypothesis testing argument (cf. \cite{tsyb}) a specific choice of functions $\sigma^2_{i,n}$ is required. For a construction see \cite{mun}, proof of Theorem 2.1 where $L:=\left(2\pi^{2\alpha}Q\right)^{1/2}/\left\|K^{\left(\alpha\right)}\right\|_{\infty}.$  It remains to show
\begin{align*}
	\sigma^2_{i,n}\in \Theta^b_s\left(\alpha,Q\right), \quad i=0,1,\ldots,M.
\end{align*}
Due to the construction of $\sigma^2_{i,n},$ we have $\sigma_{i,n}^2(t)=1$ for $t\in [0,1/4]\cup [3/4,1]$ and $\sigma_{i,n}^{2 \ \left(l\right)}\left(0\right) = \sigma_{i,n}^{2 \ \left(l\right)}\left(1\right) = 0$ for $l\in \left\{0, 1,\ldots, \alpha\right\}.$ Thus, $\sigma_{i,n}^2 \in \mathcal{W}\left(\alpha,L^2\left\|K^{\left(\alpha\right)}\right\|_\infty^2\right)$ (for a definition see Equation (\ref{eq.wdef})), $\alpha\in\left\{1,2,3,\ldots\right\}$, $j=0,\ldots,M$.  Hence by Theorem \ref{thm.equiv} it follows $\sigma^2_{i,n} \in \Theta_s\left(\alpha,Q\right)$ for $i=0,\ldots,M$.
\end{proof}

\section{Simulations}
\label{sec.sim}

In this section we briefly illustrate the performance of our estimators. Our aim is not to give a comprehensive simulation study, rather we would like to illustrate the behaviour of the estimator when assumptions of Theorems \ref{thm.misetthm} and \ref{thm.misesthm} are violated. In the following we plotted our estimator to simulated data, where we always set $n=25.000.$ From the point of view of financial statistics this is approximately the sample size obtained over a trading day ($6.5$ hours) if log-returns are sampled at every second. For simplicity, we will choose $N$ in (\ref{eq.tsdef}) and (\ref{eq.ssdef}) as the minimizer of $\left\|\hat \tau^2-\tau^2\right\|_n^2$ and $\left\|\hat \sigma^2-\sigma^2\right\|_n^2$, respectively, which is in practice unknown. Of course, proper selection of the threshold $N^*$ is of major importance for the performance of the estimator. To this end various methods are available, among others, cross validation techniques, balancing principles, and variants thereof could be employed (see e.g. \cite{DG04}, \cite{DMR96}, \cite{N08} and \cite{PS05}). A thorough investigation is postponed to a separate paper. Throughout our simulations we assumed $\tau=0.01$ and concentrated mainly on estimation of $\sigma^2,$ as it is the more challenging task.

In Figure \ref{fig:s1p} we have displayed the estimator for $\sigma(t)=\left(2+\cos\left(2\pi t\right)\right)^{1/2}$. Note that by Definition \ref{def.ssobo}, $\sigma^2$ has "infinite" smoothness, i.e. for any $\alpha>0$, we can find a $Q<\infty$, such that $\sigma^2\in \Theta_s\left(\alpha,Q\right).$ The reconstruction shows that estimation of $\tau^2$ can be done much easier than estimation of $\sigma^2$ although it is of smaller magnitude. In Figure \ref{fig:s1t2p}, we are interested in the behavior of the estimators if heavy-tailed microstructure noise is present. This was simulated by generating $\epsilon_{i,n}\sim 3^{-1/2} t\left(3\right)$, $i=1,\ldots,n$, i.i.d., where $t\left(3\right)$ denotes a t-distribution with $3$ degrees of freedom. We can see from Plot $1$ in Figure \ref{fig:s1t2p} that the resulting microstructure noise has some severe outliers according to the tail $x^{-4}$ of the density of $t(3)$. Nevertheless, estimation of $\tau^2$ and $\sigma^2$ is not visibly affected by the distribution of the noise. 

In the subsequent figures we illustrate the behaviour of the estimator when the required smoothness assumptions on $\sigma^2$ and $\tau^2$ are violated. To this end, we investigate in Figure \ref{fig:s3p} the situation when $\sigma$ is random itself, i.e. a realization from a Brownian motion, $\sigma\left(t\right)=3\left|\tilde W_t\right|.$ The Brownian motion $\left(\tilde W_t\right)_{t\in \left[0,1\right]}$ was modelled as independent from the Brownian motion in (\ref{eq.mod2}) and the microstructure noise process. It is of course not possible to reconstruct the complete path of $\sigma^2$, but as Figure \ref{fig:s3p} indicates, the estimators at least detects the smoothed shape of the path and so our estimator might already reveal some parts of the pattern of volatility also in case $\sigma$ is non-deterministic, which is certainly more realistic in most applications. 

Finally, in Figure \ref{fig:s5p} we investigated the case of $\sigma$ being a jump-function. We put $\sigma\left(t\right)=1+\mathbb{I}_{\left(\right. 1/2,1\left. \right]}\left(t\right),$ a function with jump at $t=1/2.$ Fourier series usually show a Gibbs phenomenon, i.e. an oscillating behavior at discontinuities. This behavior is also clearly visible in the graph of $\hat \sigma^2.$ In order to reconstruct jumps in volatility other methods certainly will be more suitable and are postponed to a separate paper.

{\bf Computational tasks:} We implemented the estimators in Matlab using the routine fft() for the discrete sine transform (see Remark \ref{rem.timing}). Calculation of the estimators for a sample size of $n=25.000$ took around $2$-$3$ seconds on a Intel Celeron 1.7 GHz processor. As mentioned in Remark \ref{rem.timing}, the estimator can be calculated in $O\left(Nn\log n\right)$ steps. If we choose $N$ with the optimal scale, i.e. $N\sim n^{1/\left(4\alpha+2\right)}$, we have for the complexity $O\left(Nn\log n\right)=o\left(n^{5/4}\log n\right),$ whenever $\alpha>1/2.$

\begin{appendices}

\section{Convergence rate of $\hat \sigma^2$}
\label{sec.cvr}

In this section we will give a proof of Theorem \ref{thm.misesthm}. To this end we first introduce some notation and then prove a Lemma in order to get uniform estimates of bias and variance of the single estimators $\hat s_{k,0}$.

\subsection{Preliminary Results and Notation}
\label{sec.preandnot}
Proofs of the upper bounds are based on a decomposition of $\Delta Y^k$. In this subsection we present some further notation. Let $\sigma_k(t):=\sigma(t)f_k(t)$ and $\tau_k(t):=\tau(t)f_k(t),$ $t\in [0,1].$ Let throughout the following for the Sobolev {\it s}-ellipsoids in Definition \ref{def.ssobo} for $\sigma^2$ the constants being $l=\sigma_{\min}$ and $u=\sigma_{\max}$ and for $\tau^2$, $l=\tau_{\min}$, $u=\tau_{\max}$. We define
\begin{align}
	\phi_{n}
	&:= \sup_{\sigma^2\in \Theta^b_s\left(\alpha,Q\right)}
	\max_{i=1,\ldots,n-1}\sup_{\xi \in \left[i/n,\left(i+1\right)/n\right]}
	\left|\sigma\left(\xi\right)-\sigma\left(\frac in\right)\right|,
	\notag \\
	\bar \phi_{n}
	&:=  \sup_{\tau^2\in \Theta^b_s\left(\beta,\bar Q\right)}
		\max_{k \leq n^{1/4}}
		\max_{i=1,\ldots,n}\left|\Delta_i \tau_k\right|.
	\label{eq.maxdis}
\end{align}
In order to do the proofs for model (\ref{eq.mod2}) and model (\ref{eq.mod}) simultaneously, we first define the more general process
\begin{eqnarray}
	V_{k,l}
	:=X_{1,k}+X_{2,k}+Z_{1,k,l}
	+Z_{2,k}, \quad l=1,2,
	\label{eq.deltay}
\end{eqnarray}
where $X_{1,k}, X_{2,k}, Z_{1,k,l}$ and $Z_{2,k}$ are $n-1$ dimensional random vectors with components
\begin{eqnarray*}
	\left(X_{1,k}\right)_i
	&:=& \sigma_k\left( i/n\right)\Delta_i W_{i/n} , \\
	\left(X_{2,k}\right)_i
	&:=& \tau_k\left( i/n\right)\Delta_i \epsilon_{i,n}, \\
	\left(Z_{1,k,1}\right)_i
	&:=& f_k\left(i/n\right)\int_{i/n}^{\left(i+1\right)/n}
		\left(\sigma\left(s\right)-
		\sigma\left(\frac in\right)\right)dW_s, \\
	\left(Z_{1,k,2}\right)_i
	&:=& f_k\left(i/n\right)\left(\Delta_i \sigma\right) W_{\left(i+1\right)/n}, \\
	\left(Z_{2,k}\right)_i
	&:=&
	f_k\left(i/n\right)\left(\Delta_i \tau\right)\epsilon_{i+1,n},
	\quad \quad i=1,\ldots,n-1.
\end{eqnarray*}
Obviously, $\Delta Y^k=V_{k,1}$ and $\Delta Y^k=V_{k,2}$ if model (\ref{eq.mod2}) and (\ref{eq.mod}) holds, respectively. Define the generalized estimators $\hat t_{k,0,l}:=V_{k,l}^tDJ_n^\tau D^tV_{k,l}$ and $\hat s_{k,0,l}:=V_{k,l}^tDJ_n D^tV_{k,l}-7\pi^2\hat t_{k,0,l}/3$. Further there exists a decomposition with $C_{1,k,l},C_{2,k}\in \mathbb{M}_{n-1,n}$ such that
\begin{eqnarray}
	V_{k,l} =C_{1,k,l}\xi+C_{2,k}\epsilon,
	\label{eq.deltayk}
\end{eqnarray}
where $\epsilon=\left(\epsilon_{1,n},\ldots,\epsilon_{n,n}\right)^t$ and $\xi=\xi_n$ is standard $n$-variate normal, $\epsilon, \xi$ independent and $C_{1,k,l}\xi=X_{1,k}+Z_{1,k,l}$, $C_{2,k}\epsilon=X_{2,k}+Z_{2,k}$. Now, let
\begin{eqnarray}
	s_{k,p}:=\int_0^1 \sigma^2_k(x)\cos(p\pi x)dx 
	,\quad 
	t_{k,p}:=\int_0^1 \tau^2_k(x)\cos(p\pi x)dx
	\label{def.tkl}
	\label{def.skl}
\end{eqnarray}
be the scaled $p$-th Fourier coefficients of the cosine series of $\sigma_k^2$ and $\tau_k^2$, respectively. Define the sums $A(\sigma_k^2,r)$ by
\begin{align*}
 A\left(\sigma^2_k,r\right)
	=
	\begin{cases}
	 s_{k,0}+2\sum_{m=1}^\infty s_{k,2nm} \quad &\text{for} \ r \equiv 0 \bmod 2n, \\
	\sum_{m=0}^\infty s_{k,2nm+n} \quad &\text{for} \ r \equiv n \bmod 2n, \\
	\sum_{q \equiv \pm r \bmod 2n, \ q\geq 0} s_{k,q} \quad &\text{for} \ r \not \equiv 0 \bmod n,
	\end{cases}
\end{align*}
and analogously $A(\tau_k^2,r)$ with $s_{k,p}$ replaced by $t_{k,p}$. Some properties of these variables are given in Lemma \ref{lem.approx} and Lemma \ref{lem.bound}. 

Further define
\begin{eqnarray}
	\Sigma_k:=
	\left(
	\begin{array}{ccc}
	\sigma_k(1/n) & & \\
	& \ddots & \\
	& & \sigma_k(1-1/n)
	\end{array}
	\right).
	\label{eq.sigmak}
\end{eqnarray}
We put $\Cum_4\left(\epsilon\right):=\Cum_4\left(\epsilon_{1,n}\right)$ for the fourth cumulant of $\epsilon_{1,n}$. If $X,Y$ are independent random vectors, we write $X \perp Y$.

\subsection{Proofs for Estimation of $\sigma^2$}
\label{subsec.estofsigma}

\begin{lem2}
\label{lem.sklem}
Let $\hat s_{k,0}$ be defined as in (\ref{eq.defsk}). Further assume $\beta>1$, $Q, \bar Q>0,$ $0<\sigma_{\min}\leq \sigma_{\max}<\infty,$ $0<\tau_{\min}\leq \tau_{\max}<\infty$ and $k=k_n \in \mathbb{N}$.
\begin{itemize}
\item[(i)] Assume model (\ref{eq.mod2}), $\alpha>1/2$. Then it holds
\begin{align}
	\sup_{\sigma^2\in \Theta^b_s\left(\alpha,Q\right), \ \tau^2\in \Theta^b_s\left(\beta,\bar Q\right)}
	&\max_{k\leq n^{1/4}}\left|\E\left(\hat s_{k,0}\right)-s_{k,0}\right|
	=O\left(n^{-1/4}+
	n^{1-\beta}+n^{1/2-\alpha}\right),
	\label{eq.skE2}\\
	\sup_{\sigma^2\in \Theta^b_s\left(\alpha,Q\right), \ \tau^2\in \Theta^b_s\left(\beta,\bar Q\right)}
	&\max_{k\leq n^{1/4}}\Var\left(\hat s_{k,0}\right)
	=O\left(n^{-1/2}+n^{4-4\beta}\right).
	\label{eq.skV2}
\end{align}
\item[(ii)] Assume model (\ref{eq.mod}), $\alpha>5/4$. Then it holds
\begin{align}
	\sup_{\sigma^2\in \Theta^b_s\left(\alpha,Q\right), \ \tau^2\in \Theta^b_s\left(\beta,\bar Q\right)}
	&\max_{k\leq n^{1/4}}\left|\E\left(\hat s_{k,0}\right)-s_{k,0}\right|
	=
	O\left(n^{1-\beta}+n^{5/2-2\alpha}+n^{-1/4}\right),
	\label{eq.skE}\\
	\sup_{\sigma^2\in \Theta^b_s\left(\alpha,Q\right), \ \tau^2\in \Theta^b_s\left(\beta,\bar Q\right)}
	&\max_{k\leq n^{1/4}}\Var\left(\hat s_{k,0}\right)
	=O\left(n^{-1/2}+n^{4-4\beta}\right).
	\label{eq.skV}
\end{align}
\end{itemize}
\end{lem2}

\begin{proof}
The proof mainly uses the generalized estimators as introduced in Section \ref{sec.preandnot}. 
It is clear that for two centered random vectors $P$ and $Q$ 
\begin{eqnarray*}
	\left\langle P, Q\right\rangle_\sigma :=
	\E\left(P^tDJ_n D Q\right)
\end{eqnarray*}
defines a semi-inner product and by Lemma \ref{lem.n}, $P \perp Q \Rightarrow \left\langle P, Q\right\rangle_\sigma =0$.
Hence
\begin{eqnarray}
	\E\hat s_{k,0,l}
	&=&
	\left\langle X_{1,k}, X_{1,k} \right\rangle_\sigma 
	+ \left\langle X_{2,k}, X_{2,k} \right\rangle_\sigma 
	+ \left\langle Z_{1,k,l}, Z_{1,k,l} \right\rangle_\sigma 
	+ \left\langle Z_{2,k}, Z_{2,k} \right\rangle_\sigma \notag \\
	&& \quad 
	+2 \left\langle X_{1,k}, Z_{1,k,l} \right\rangle_\sigma 
	+2 \left\langle X_{2,k}, Z_{2,k} \right\rangle_\sigma
	-\frac{7\pi^2}3 
	\E\left(\hat t_{k,0,l}\right)
	.
	\label{eq.evalex}
\end{eqnarray}
Clearly with (iii) in Lemma \ref{lem.approx} and $r_n:=n^{-1/2}\left[n^{1/2}\right]$,
\begin{eqnarray*}
	\left\langle X_{1,k}, X_{1,k} \right\rangle_\sigma 
	&=& \frac 1n \Tr\left(\Sigma_k D J_n D\Sigma_k\right)
	=\frac 1n\Tr\left(J_nD\Sigma_k^2D\right) \\
	&=&
	n^{-1/2}\sum_{i=\left[n^{1/2}\right]+1}^{2\left[n^{1/2}\right]}
	\left(
	A\left(\sigma_k^2,0\right)
	-
	A\left(\sigma_k^2,2i\right)
	\right) \\
	&=& r_n A\left(\sigma_k^2,0\right)
	-n^{-1/2} \sum_{i=\left[n^{1/2}\right]+1}^{2\left[n^{1/2}\right]}
	A\left(\sigma_k^2,2i\right).
\end{eqnarray*}
Hence due to $r_n\leq 1$ and $\left|r_n-1\right|\leq n^{-1/2}$
\begin{eqnarray*}
	\left| \left\langle X_{1,k}, X_{1,k}\right\rangle_\sigma -s_{k,0}\right|
	\leq 2\sum_{m=n}^\infty \left|s_{k,m}\right|
	+\frac 2{\sqrt{n}}\sum_{i=0}^\infty  \left|s_{k,i}\right|,
\end{eqnarray*}
and  with Lemma \ref{lem.bound}
\begin{align*}
	\sup_{\sigma^2\in \Theta_s^b\left(\alpha, Q\right)}
	\max_{k\leq n^{1/4}} 
	\left| \left\langle X_{1,k}, X_{1,k}\right\rangle_\sigma -s_{k,0}\right|=O\left(n^{1/2-\alpha}\right).
\end{align*}
Next we will bound $\left\langle X_{2,k},X_{2,k}\right\rangle_{\sigma}.$ In order to do this let $T_k \in \mathbb{D}_{n-1}$ with entries $\left(T_k\right)_{i,j}=\tau_k\left(i/n\right)\delta_{i,j}.$ Further we define $\tilde T_k \in \mathbb{M}_{n-1}$
\begin{align}
	\left(\tilde T_k\right)_{i,j}
	=
	\begin{cases}
	\left(\Delta_i \tau_k\right)^2 \quad &\text{for} \quad i=j-1, \\
	\left(\Delta_j \tau_k\right)^2 \quad &\text{for} \quad i=j+1, \\
	0 \quad &\text{otherwise}.
	\end{cases}
	\label{eq.ttildedef}
\end{align}
Note the relation 
\begin{eqnarray}
	\Cov\left(X_{2,k}\right)=T_kKT_k=1/2 T_k^2K+1/2KT_k^2+1/2\tilde T_k.
	\label{eq.tktdec}
\end{eqnarray}
Using Lemma \ref{lem.p} yields
\begin{align}
	\left\langle X_{2,k}, X_{2,k} \right\rangle_\sigma
	&=\E\left(X_{2,k}^t T_k D J_n D X_{2,k}\right)
	=\Tr\left(D J_n D T_k KT_k\right) \notag  \\
	&= \frac 12 \Tr\left(J_n D T_k^2 K D \right)
	+ \frac 12 \Tr\left(J_n D K T_k^2 D \right)
	+ \frac 12 \Tr\left(J_n D \tilde T_k D\right) \notag \\
	&= \Tr\left(\Lambda  J_n D T_k^2 D\right)
	+ \frac 12 \Tr\left(J_n D \tilde T_k D\right),
	\label{eq.x2x2}
\end{align}
and further
\begin{align}
	\Tr\left(\Lambda  J_n D T_k^2 D\right)
	&= n^{1/2} \sum_{i=\left[n^{1/2}\right]+1}^{2\left[n^{1/2}\right]} \lambda_i \left(A\left(\tau^2_k,0\right)-A\left(\tau^2_k,2i\right)\right) \notag \\
	&= A\left(\tau^2_k,0\right) n^{1/2} \sum_{i=\left[n^{1/2}\right]+1}^{2\left[n^{1/2}\right]}\lambda_i
	-
	n^{1/2} \sum_{i=\left[n^{1/2}\right]+1}^{2\left[n^{1/2}\right]} \lambda_i
	A\left(\tau^2_k,2i\right).
	\label{eq.trx2x2s}
\end{align}
Because $\max_{i=\left[n^{1/2}\right]+1,\ldots,2\left[n^{1/2}\right]}\lambda_i =\lambda_{2\left[n^{1/2}\right]}\leq 4\pi^2n^{-1}$, it holds
\begin{align*}
	\left|
	\sqrt n \sum_{i=\left[n^{1/2}\right]+1}^{2\left[n^{1/2}\right]} \lambda_i
	A\left(\tau^2_k,2i\right)
	\right|
	&\leq
	n^{1/2} \sum_{i=\left[n^{1/2}\right]+1}^{2\left[n^{1/2}\right]} \lambda_i
	\sum_{q \equiv \pm 2i \bmod{ 2n }, \ q\geq 0}
	\left|t_{k,q}\right| \\
	&\leq 4\pi^2n^{-1/2}\sum_{i=0}^\infty \left|t_{k,i}\right|
	\leq 8 \pi^2n^{-1/2}\sum_{i=0}^\infty \left|t_{0,i}\right|.
\end{align*}
Therefore, (\ref{eq.trx2x2s}) can be written as
\begin{align*}
	&\left|
	\Tr\left(\Lambda  J_n D T_k^2 D\right)
	-t_{k,0} n^{1/2} \sum_{i=\left[n^{1/2}\right]+1}^{2\left[n^{1/2}\right]} \lambda_i
	\right| 
	\leq 8\pi^2 \sum_{m=n}^\infty \left|t_{k,m}\right| 
	+8 \pi^2n^{-1/2}\sum_{i=0}^\infty \left|t_{0,i}\right|.
\end{align*}
This gives by Lemma \ref{lem.sumeig} and Lemma \ref{lem.bound} 
\begin{align*}
	\sup_{\tau^2\in \Theta_s^b \left(\beta, \bar Q\right)}
	&\max_{k\leq n^{1/4}} \left|
	\Tr\left(\Lambda  J_n D T_k^2 D\right)
	-\frac{7\pi^2}3 t_{k,0}
	\right| =O\left(n^{-1/2}\right).
\end{align*}
Recall that $\Tr\left(J_n\right)=O\left(n\right)$. It follows $$\left|\Tr\left(J_n D\tilde T_k D\right)\right|\leq \Tr\left(J_n\right)\max_{i,j}\left|\left(D\tilde T_k D\right)_{i,j}\right| \leq 4 \Tr\left(J_n\right) \max_i\left(\Delta_i \tau_k\right)^2.$$ So,
\begin{align}
	\sup_{\tau^2\in \Theta_s^b \left(\beta, \bar Q\right)}
	&\max_{k\leq n^{1/4}} \ \Tr\left(J_n D\tilde T_k D\right)=O\left(n\bar \phi_{n}^2\right)
	\label{eq.t2}
\end{align}
and therefore
\begin{align*}
	\sup_{\tau^2\in \Theta_s^b \left(\beta, \bar Q\right)}
	  \max_{k\leq n^{1/4}} \left| 
	\left\langle X_{2,k}, X_{2,k} \right\rangle_\sigma
	- \frac{7\pi^2}3 t_{k,0} \right|
	=
	O\left(n^{-1/2}
	+n\bar \phi_{n}^2\right)
	.
\end{align*}
We bound the remaining terms of (\ref{eq.evalex}). Note
\begin{eqnarray*}
	\left\langle Z_{1,k,1}, Z_{1,k,1}\right\rangle_\sigma
	=\Tr\left(DJ_nD \Cov\left(Z_{1,k,1}\right)\right)
	\leq
	\lambda_1\left(\Cov\left(Z_{1,k,1}\right)\right)\Tr\left(J_n\right)
	\leq 2\phi_n^2
\end{eqnarray*}
implying
\begin{eqnarray*}
	\sup_{\sigma^2\in \Theta_s^b\left(\alpha, Q\right)}
	\max_{k\leq n^{1/4}}
	\left\langle Z_{1,k,1}, Z_{1,k,1}\right\rangle_\sigma
	= O\left(\phi_{n}^2\right).
	\label{eq.zik3zik3}
\end{eqnarray*}
In order to bound $\left\langle Z_{1,k,2}, Z_{1,k,2}\right\rangle_\sigma$ define
\begin{eqnarray}
	L:=\left(\frac{\left(i\wedge j\right)+1}n\right)_{i,j=1,\ldots,n-1}
	\label{eq.Kdef}
\end{eqnarray}
and $\Delta \Sigma_k \in \mathbb{D}_{n-1}$ by
\begin{align*}
 \left(\Delta \Sigma_k\right)_{i,j} 
	:= f_k\left(\frac in\right)\left(\Delta_i \sigma\right) \delta_{i,j}.
\end{align*}
We obtain
\begin{eqnarray}
	\left\langle Z_{1,k,2}, Z_{1,k,2}\right\rangle_\sigma
	= \Tr\left(D J_n D \left(\Delta \Sigma_k\right) L \left(\Delta \Sigma_k\right)\right)
	\leq  2n^{3/2}\phi_{n}^2
\end{eqnarray}
and hence
\begin{eqnarray}
	\sup_{\sigma^2\in \Theta_s^b\left(\alpha, Q\right)}
	\max_{k\leq n^{1/4}}
	\left\langle Z_{1,k,2}, Z_{1,k,2}\right\rangle_\sigma
	= O\left(n^{3/2}\phi_{n}^2\right).
	\label{eq.zik2zik2}
\end{eqnarray}
Similarly to (\ref{eq.zik3zik3}), $\left\langle Z_{2,k}, Z_{2,k}\right\rangle_\sigma\leq \bar \phi_{n}^2 n$
and thus
\begin{align}
	\sup_{\tau^2\in \Theta_s^b\left(\beta, \bar Q\right)}
	\max_{k\leq n^{1/4}} \left\langle Z_{2,k}, Z_{2,k}\right\rangle_\sigma=O\left(\bar \phi_{n}^2 n\right).
	\label{eq.z2klz2kl}
\end{align}
Note that
\begin{align*}
	\Cov\left(X_{1,k},Z_{1,k,2}\right)_{i,j}
	=\begin{cases}
		0 \quad \quad \quad  \quad \quad \quad \quad \ & \text{for} \quad j<i \\
		n^{-1}f_k^2\left(i/n\right)
		\sigma\left(i/n\right)\left(\Delta_j\sigma\right)
		\quad  & \text{for} \quad j\geq i
	 \end{cases}
	.
\end{align*}
Hence by Proposition \ref{prop.X1R1prop}, we obtain
\begin{align*}
	&\sup_{\sigma^2\in \Theta_s^b \left(\alpha, Q\right)}\max_{k\leq n^{1/4}}
	\left| \left\langle 
	X_{1,k},Z_{1,k,2} \right\rangle _\sigma \right|
	= O\left(n^{1/2}\log n \ \phi_{n}\right)
\end{align*}
Applying the CS-inequality gives 
\begin{align*}
	&\sup_{\sigma^2\in \Theta_s^b \left(\alpha, Q\right)}
	\max_{k\leq n^{1/4}}
	\left| \left\langle
	X_{1,k},Z_{1,k,1} \right\rangle _\sigma \right|
	=O\left(\phi_n\right), \\
	&\left|
	\left\langle 
	X_{2,k},Z_{2,k} \right\rangle _\sigma
	\right|
	\leq 
	\left\langle 
	X_{2,k},X_{2,k} \right\rangle _\sigma^{1/2}
	\left\langle 
	Z_{2,k},Z_{2,k} \right\rangle _\sigma^{1/2}.
\end{align*}
Using Proposition \ref{prop.prop21} this yields (\ref{eq.skE2}) and (\ref{eq.skE}). In order to give an upper bound for the variance of $\hat s_{k,0,l}$ note
\begin{eqnarray*}
	\Var\left(\hat s_{k,0,l}\right)
	\leq
	2\Var\left(V_{k,l}^tDJ_nD V_{k,l}\right)
	+2\frac{7^2\pi^4}9\Var\left(\hat t_{k,0,l}\right).
\end{eqnarray*}
Furthermore we have using (\ref{eq.deltayk}) and Lemma \ref{lem.p} (vi)
\begin{align*}
	V_{k,l}^tDJ_nD V_{k,l}
	&=
	\xi^tC_{1,k,l}^tDJ_nDC_{1,k,l}^t \xi
	+
	2\xi^tC_{1,k,l}^tDJ_nDC_{2,k} \epsilon
	+
	\epsilon^tC_{2,k}^tDJ_nDC_{2,k} \epsilon \\
	&\leq 2\xi^tC_{1,k,l}^tDJ_nDC_{1,k,l}^t \xi
	+ 2\epsilon^tC_{2,k}^tDJ_nDC_{2,k} \epsilon.
\end{align*}
Hence
\begin{align*}
	&\Var\left(V_{k,l}^tDJ_nD V_{k,l}\right) \leq 
	8\Var\left(
	\xi^tC_{1,k,l}^tDJ_nDC_{1,k,l}\xi\right)
	+
	8\Var\left(\epsilon^tC_{2,k}^tDJ_nDC_{2,k}\epsilon\right).
\end{align*}
Finally, we bound $\Var\left(\xi^tC_{1,k,l}^tDJ_nDC_{1,k,l}\xi\right)$ and $\Var\left(\epsilon^tC_{2,k}^tDJ_nDC_{2,k}\epsilon\right)$ in two steps, which will be denoted by $(a)$ and $(b)$.

\medskip

\noindent
(a) By Lemma \ref{lem.ABineq} (iii), we have
\begin{align}
	\Var&\left(
	\xi^tC_{1,k,l}^tDJ_nDC_{1,k,l}\xi\right)
	=
	2\left\|
	J_n^{1/2}D\Cov\left(X_{1,k}+Z_{1,k,l}\right)DJ_n^{1/2}\right\|_F^2 \notag \\
	& \quad  \leq
	8
	\left\|
	J_n^{1/2}D\left(\Cov\left(X_{1,k}\right)
	+\Cov\left(Z_{1,k,l}\right) \right)DJ_n^{1/2}\right\|_F^2 \notag \\
	&\quad \leq
	16
	\left\|
	J_n^{1/2}D\Cov\left(X_{1,k}\right)DJ_n^{1/2}\right\|_F^2
	+16
	\left\|
	J_n^{1/2}D\Cov\left(Z_{1,k,l}\right)DJ_n^{1/2}\right\|_F^2.
	\label{eq.varub}
\end{align}
Firstly,
\begin{align*}
	\left\|J_n^{1/2}D\Cov\left(Z_{1,k,1}\right)DJ_n^{1/2}\right\|_F^2
	&\leq
	\lambda_1^2\left(\Cov\left(Z_{1,k,1}\right)\right)
	\Tr\left(J_n^2\right)
	\leq 4n^{-1/2}\phi_n^4,\\
	\left\|J_n^{1/2}D\Cov\left(Z_{1,k,2}\right)DJ_n^{1/2}\right\|_F^2
	&\leq
	\lambda_1^2\left(
	DJ_nD
	\right)
	\Tr\left(
	\Cov\left(Z_{1,k,2}\right)^2
	\right) \leq
	4n\phi_{n}^4 \left\|L\right\|_F^2
	\leq 4n^3\phi_{n}^4,
\end{align*}
and also
\begin{eqnarray*}
	\left\|
	J_n^{1/2}D\Cov\left(X_{1,k}\right)DJ_n^{1/2}\right\|_F^2
	\leq
	\lambda_1^2\left(\Cov\left(X_{1,k}\right)\right)\Tr\left(J_n^2\right)\leq \sigma_{\max}n^{-1/2}.
\end{eqnarray*}
Therefore,
\begin{eqnarray*}
	\sup_{\sigma^2\in \Theta_s^b \left(\alpha, Q\right)}
	\max_{k\leq n^{1/4}} \ 
	\Var\left(
	\xi^tC_{1,k,l}^tDJ_nDC_{1,k,l}\xi\right)
	= O\left(n^{-1/2}\right).
\end{eqnarray*}

\medskip

\noindent
(b) Next, we see with the same arguments as in (\ref{eq.varub})
\begin{align*}
	\Var\left(\epsilon^tC_{2,k}^tDJ_nDC_{2,k}\epsilon\right)
	&\leq \left(2+\Cum_4\left(\epsilon\right)\right) 
	\left\|
	J_n^{1/2}D\Cov\left(X_{2,k}+Z_{2,k}\right)DJ_n^{1/2}\right\|_F^2 \\
	&\leq 
	8\left(2+\Cum_4\left(\epsilon\right)\right) \left\|
	J_n^{1/2}D\Cov\left(X_{2,k}\right)DJ_n^{1/2}\right\|_F^2 \\
	&\quad  +
	8\left(2+\Cum_4\left(\epsilon\right)\right)
	\left\|
	J_n^{1/2}D\Cov\left(Z_{2,k}\right)DJ_n^{1/2}\right\|_F^2.
\end{align*}
We obtain
\begin{eqnarray*}
	\left\|
	J_n^{1/2}D\Cov\left(Z_{2,k}
	\right)DJ_n^{1/2}\right\|_F^2
	\leq 4\bar \phi_{n}^4 \Tr\left(J_n^2\right)
	=4\bar \phi_{n}^4 n^{3/2}.
\end{eqnarray*}
From (\ref{eq.ttildedef}) follows now
\begin{eqnarray*}
	\left\|
	J_n^{1/2}D\Cov\left(X_{2,k}\right)DJ_n^{1/2}\right\|_F^2 
	\leq
	\frac 32 
	\left\|
	J_n^{1/2}DT_k^2KDJ_n^{1/2}\right\|_F^2
	+\frac 34
	\left\|
	J_n^{1/2}D\tilde T_kDJ_n^{1/2}\right\|_F^2.
\end{eqnarray*}
Let $I_{n-1}$ be the $n-1\times n-1$ identity matrix. Due to $\Lambda J_n \Lambda \leq I_n \lambda^2_{2\left[n^{1/2}\right]}n^{1/2}$ we have
\begin{align*}
	\left\|
	J_n^{1/2}DT_k^2KDJ_n^{1/2}\right\|_F^2
	&=
	\Tr\left( J_n^{1/2}DT_k^2 \Lambda J_n\Lambda T_k^2DJ_n^{1/2}\right) \\
	&\leq
	\lambda^2_{2\left[n^{1/2}\right]} n^{1/2}\Tr\left(J_n^{1/2}DT_k^4DJ_n^{1/2}\right) \\
	&\leq
	\lambda^2_{2\left[n^{1/2}\right]} n^{1/2} \tau_{\max}^2\Tr\left(J_n\right).
\end{align*}
Also
\begin{eqnarray*}
	\left\|
	J_n^{1/2}D\tilde T_kDJ_n^{1/2}\right\|_F^2
	\leq \lambda_1^2\left(J_n\right)
	\left\|\tilde T_k\right\|_F^2
	\leq 2n^2\bar \phi_{n}^4
\end{eqnarray*}
and therefore
\begin{eqnarray*}
	\sup_{\tau^2\in \Theta_s^b \left(\beta, \bar Q\right)}\max_{k\leq n^{1/4}}
	\ 
	\Var\left(
	\epsilon^tC_{2,k}^tDJ_nD C_{2,k}\epsilon
	\right)=
	O\left(n^{-1/2}+n^2\bar \phi_{n}^4\right).
\end{eqnarray*}

\medskip

\noindent
Combining (a) and (b) gives
\begin{align*}
	\sup_{\sigma^2\in \Theta^b_s\left(\alpha,Q\right), \ \tau^2\in \Theta^b_s\left(\beta,\bar Q\right)}
	&\max_{k\leq n^{1/4}} \ 
	\Var\left(
	V_{k,l}^tDJ_nDV_{k,l}
	\right)
	=O\left(n^{-1/2}+n^2\bar \phi_{n}^4\right),
\end{align*}
so (\ref{eq.skV2}) and (\ref{eq.skV}) follow with Lemma \ref{lem.tlem} and Propositon \ref{prop.prop21}.
\end{proof}

\begin{proof}[Proof of Theorem \ref{thm.misesthm}]
We decompose $$\MISE\left(\hat \sigma^2_N\right)=\int_0^1\Bias^2\left(\hat \sigma^2_N(t)\right)
dt+ \int_0^1 \Var\left(\hat \sigma^2_N(t)\right)dt.$$ We have that $\sigma^2(t)
=\sum_{i=0}^\infty \left\langle \psi_k,\sigma^2\right\rangle \psi_k\left(t\right),$ where $\left\langle \ . \ , \ . \ \right\rangle$ denotes the standard scalar product on $L^2[0,1]$. Let $\eta_{k,n,l}:=\E\left(\hat s_{k,0,l}\right)-s_{k,0}$. Then for $i\geq 1$, $\E\left(2\hat s_{i,0,l}-2\hat s_{0,0,l}\right)=2^{1/2}\left\langle \psi_i,\sigma^2\right\rangle+2\eta_{i,n,l}-2\eta_{0,n,l}$. Hence
\begin{eqnarray*}
	\int_0^1 \Bias^2\left(\hat \sigma^2(t)\right) dt
	&=& \eta_{0,n,l}^2
	+2\sum_{i=1}^N \left(\eta_{i,n,l}
		-\eta_{0,n,l}\right)^2
	+\sum_{i=N+1}^\infty \left\langle \psi_i,\sigma^2\right\rangle^2
\end{eqnarray*}
and we obtain
\begin{eqnarray*}
	\eta_{0,n,l}^2
	+2\sum_{i=1}^N \left(\eta_{i,n,l}
	-\eta_{0,n,l}\right)^2
	&\leq& \left(8N+1\right)\max_{i=0,\ldots,N} 
	\eta_{i,n,l}^2.
\end{eqnarray*}
Because $\sigma^2 \in \Theta_s\left(\alpha, Q \right)$ it  holds
\begin{align}
	\sum_{i=N+1}^\infty \left\langle\psi_i,\sigma^2\right\rangle^2
	&\leq \frac 1{\left(N+1\right)^{2\alpha}}\sum_{i=1}^ \infty
	i^{2\alpha}\left\langle\psi_i,\sigma^2\right\rangle^2
	\leq 2Q\left(N+1\right)^{-2\alpha}.
	\label{eq.sumfcoeff}
\end{align}
Therefore,
\begin{align*}
	&\sup_{\sigma^2\in \Theta^b_s\left(\alpha,Q\right), \ \tau^2\in \Theta^b_s\left(\beta,\bar Q\right)}
	\int_0^1 \Bias^2\left(\hat \sigma^2(t)\right) dt \\
	& \quad =O\left(N \sup_{\sigma^2\in \Theta^b_s\left(\alpha,Q\right), \  \tau^2\in \Theta^b_s\left(\beta,\bar Q\right)}
	\max_{i=0,\ldots,N} \eta_{i,n,l}^2+ N^{-2\alpha}\right).
\end{align*}
Assume model (\ref{eq.mod2}), then by Lemma \ref{lem.sklem} 
\begin{align*}
	\sup_{\sigma^2\in \Theta^b_s\left(\alpha,Q\right), \tau^2\in \Theta^b_s\left(\beta,\bar Q\right)}
	&\int_0^1 \Bias^2\left(\hat \sigma^2(t)\right) dt \\ &=O\left(Nn^{-1/2}+Nn^{2-2\beta}+Nn^{1-2\alpha}+N^{-2\alpha}\right).
\end{align*}
and for model (\ref{eq.mod}),
\begin{align*}
	\sup_{\sigma^2\in \Theta^b_s\left(\alpha,Q\right), \ \tau^2\in \Theta^b_s\left(\beta,\bar Q\right)}
	&\int_0^1 \Bias^2\left(\hat \sigma^2(t)\right) dt \\
	&= O\left(Nn^{2-2\beta}+Nn^{5-4\alpha}+Nn^{-1/2}+N^{-2\alpha}\right).
\end{align*}
For the variance note
\begin{align*}
	\int_0^1 \Var\left(\hat \sigma^2\left(t\right)\right) dt
	&= \Var\left(\hat s_{0,0,l}\right)
	+\frac 12\sum_{i=1}^N\Var\left(2\hat s_{i,0,l}-2\hat s_{0,0,l}\right) \\
	&\leq \left(4N+1\right)\Var\left(\hat s_{0,0,l}\right)
	+4\sum_{i=1}^N \Var\left(\hat s_{i,0,l}\right).
\end{align*}
Hence
\begin{align*}
	\sup_{\sigma^2\in \Theta^b_s\left(\alpha,Q\right), \ \tau^2\in \Theta^b_s\left(\beta,\bar Q\right)}
	&\int_0^1 \Var\left(\hat \sigma^2\left(t\right)\right) dt \\
	&=
	O\left(N
	\sup_{\sigma^2\in \Theta^b_s\left(\alpha,Q\right), \tau^2\in \Theta^b_s\left(\beta,\bar Q\right)}
	\max_{0\leq k\leq N}\Var\left(\hat s_{k,0,l}\right)\right).
\end{align*}
Using Lemma \ref{lem.sklem} yields the result.
\end{proof}

\section{Sobolev s-ellipsoids}
\label{sec.ass}

In this chapter we will shortly discuss the function space introduced in Section \ref{sec.techpre} and provide a theorem needed for the lower bound. First recall the classical definition of Sobolev ellipsoids (cf. Proposition 1.14 in \cite{tsyb}).
\begin{defi2}
\label{def.sobell}
Define
\begin{eqnarray*}
	a_{j,\alpha}=
	\begin{cases}
	j^\alpha, \quad \text{for even} \ j, \\
	\left(j-1\right)^\alpha, \quad \text{for odd} \ j
	\end{cases}.
\end{eqnarray*}
Let $\{\varphi_j\}_{j=1}^\infty$, $\phi_1(x):=1$, $\phi_{2j}(x):=\sqrt{2}\cos\left(2\pi jx\right),$ $\phi_{2j+1}(x):=\sqrt{2}\sin\left(2\pi jx\right)$ denote the trigonometric basis on $[0,1]$. Then we call the function space
\begin{align*}
	&\Theta\left(\alpha,C\right)
	:= 
	 \left\{
	f \in L^2[0,1]: \exists \left(\theta_n\right)_{n=1}^\infty, \ \text{s. t.} \ f(x)=\sum_{i=1}^\infty \theta_i\varphi_i\left(x\right),\sum_{i=1}^\infty a_{i,\alpha}^2\theta_i^2 \leq C
	\right\}
\end{align*}
a Sobolev ellipsoid.
\end{defi2}
Interesting characterizations arise if we put Sobolev {\it s}-ellipsoids into relation with Sobolev ellipsoids:
\begin{rem2}
Let $\mathcal{S}$ be the class of all symmetric functions $f\in L^2[0,1]$ such that $f(x)=f(1-x)$, $\forall x\in[0,1]$. Further let $\Theta(\alpha,C)$ be a Sobolev ellipsoid. Then a function belongs to $\Theta_s(\alpha,C)\cap \mathcal{S}$ if and only if it belongs to $\Theta(\alpha,C)\cap \mathcal{S}$.
\end{rem2}

\noindent
Let
\begin{align}
	\mathcal{W}(\alpha,\bar C):=&\mathcal{W}(\alpha,\bar C,[0,1])
	:=  \left\{f \in L^2[0,1]: f^{\left(l\right)}(0)=f^{\left(l\right)}(1)=0
	\right. \notag \\ & \left.\quad 
	\ \text{for} \ l \ \text{odd}, \ l< \alpha, \int_0^1\left(f^{\left(\alpha\right)}(x)\right)^2dx \leq \bar C\right\}.
	\label{eq.wdef}
\end{align}
For positive integer values of $\alpha$, we have the following equivalence.
\begin{thm2}
\label{thm.equiv}
Assume $\alpha \in \left\{1,2,\ldots\right\}$, $C>0$. Let $\bar C=2\pi^{2\alpha}C$. Then a function is in $\mathcal{W}(\alpha,\bar C)$ if and only if it is in $\Theta_s(\alpha,C)$.
\end{thm2}

\begin{proof}
First we show that if a function $f \in \mathcal{W}(\alpha,\bar C)$ then also $f\in \Theta_s(\alpha,C)$. Let $\tilde f$ be defined on $[-1,1]$ by
\begin{align*}
	\tilde f(x):=
	\begin{cases}
	f(x) \quad \text{for} \quad & x\in [0,1] \\
	f(-x) \quad \text{for} \quad & x\in [-1,0] \\
	\end{cases}
	.
\end{align*}
Note that $\tilde f$ is an $\alpha$-times differentiable function,  $\tilde f^{\left(l\right)}$ is even if $l$ is even and $\tilde f^{\left(l\right)}$ is odd if $l$ is odd. Let 
\begin{align*}
	s_k(j)=
	\begin{cases}
	\int_{-1}^1\tilde f^{\left(j\right)}(x)dx \quad &\text{for} \quad k=0 \\
	\int_{-1}^1 \tilde f^{\left(j\right)}(x)\cos(k\pi x)dx 
		\quad &\text{for} \quad k\geq1, j \ \text{even} \\
	\int_{-1}^1 \tilde f^{\left(j\right)}(x)\sin(k\pi x)dx 
		\quad &\text{for} \quad k\geq1, j \ \text{odd}
	\end{cases}
	.
\end{align*}
It holds for $j\geq 1$
\begin{eqnarray*}
	s_0(j)=\int_{-1}^1 \tilde f^{\left(j\right)}(x)dx
	= \tilde f^{\left(j-1\right)}(1)-\tilde f^{\left(j-1\right)}(-1)=0.
\end{eqnarray*}
Hence we have the Parseval type equality
\begin{eqnarray}
	\int_0^1 \left(f^{\left(\alpha\right)}(x)\right)^2 dx
	= \frac 12\sum_{k=1}^\infty s_k^2(\alpha).
	\label{eq.par}
\end{eqnarray}
Further for $k\geq 1$, $j$ even, it follows by partial integration
\begin{align*}
	s_k(j)&=\int_{-1}^1 \tilde f^{\left(j\right)}(x)\cos(k\pi x)dx \\
	&=\left. \tilde f^{\left(j-1\right)}(x)\cos(k\pi x)\right|_{-1}^1
	+k\pi\int_{-1}^1\tilde f^{\left(j-1\right)}(x)\sin(k\pi x)dx 
	= k\pi s_k(j-1)
\end{align*}
and for $k\geq1$ and $j$ odd
\begin{align*}
	s_k(j)
	&=\int_{-1}^1 \tilde f^{\left(j\right)}(x)\sin(k\pi x)dx \\
	&=\left. \tilde f^{\left(j-1\right)}(x)\sin(k\pi x)\right|_{-1}^1
	-k\pi\int_{-1}^1\tilde f^{\left(j-1\right)}(x)\cos(k\pi x)dx
	= -k\pi s_k(j-1).
\end{align*}
With $\theta_k=\int_0^1 f(x)\cos(k\pi x)dx$ it follows for $k\geq 1$ $s_k^2(\alpha)=k^{2\alpha}\pi^{2\alpha}s_k^2=4k^{2\alpha}\pi^{2\alpha}\theta_k^2$. Combining this result with (\ref{eq.par}) yields
\begin{eqnarray*}
	\int_0^1 \left(f^{\left(\alpha\right)}(x)\right)^2 dx
	=2\pi^{2\alpha}\sum_{k=1}^\infty k^{2\alpha}\theta_k^2
\end{eqnarray*}
and hence proves the first part of the theorem. The other direction follows in a straightforward way by differentiation and is thus omitted.
\end{proof}

\section*{Supplementary Material}
{\bf Supplement: Proofs for upper bound of $\hat \tau^2_N$ and further technicalities}
\newline
(http://www.stochastik.math.uni-goettingen.de/munk) In the supplementary material we provide a proof for Theorem \ref{thm.misetthm} and summarize results from linear algebra and matrix theory needed for the proofs.

\section*{Acknowledgments}
We are grateful to T. Tony Cai, Marc Hoffmann, Mark Podolskij and Ingo Witt for helpful comments and discussions.

\bibliographystyle{plain}       
\bibliography{refsPart1}           

\newpage
\begin{figure}
	\begin{center}
	\includegraphics[width=10cm]{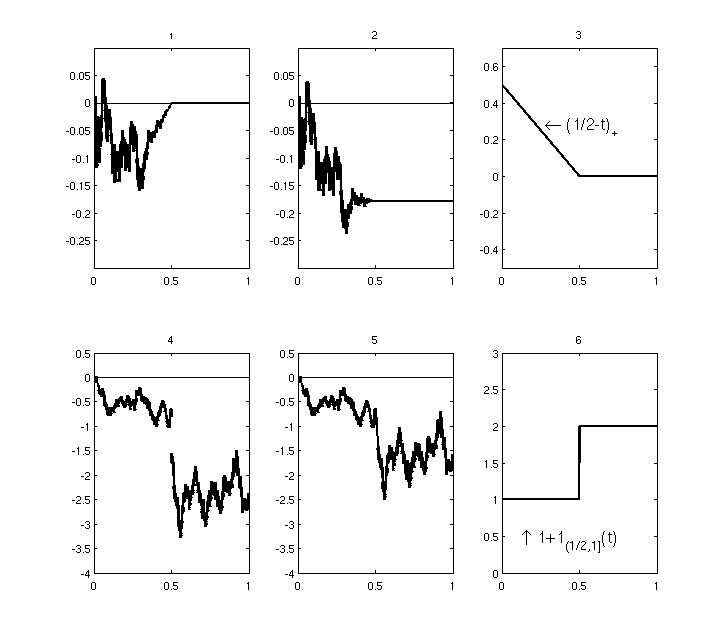}
	\end{center}
	\caption{Plots $1$ and $2$ display paths of sBM and tBM corresponding to  $\sigma(t)=\left(1/2-t\right)_+$ (Plot $3$). Analogously, Plots $4$ and $5$ show paths of sBM and tBM with $\sigma\left(t\right)=1+\mathbb{I}_{\left(\right. 1/2,1\left. \right]}\left(t\right)$ (Plot $6$). For Plots $1$ and $2$ as well as Plots $4$ and $5$ we took the same realization $\left(W_t\right)_{t\in \left[0,1\right]}$ of the underlying Brownian motion. The first two plots show the different scaling behavior: sBM$=0$ and tBM$=\int_0^{1/2}\sigma\left(s\right)dW_s$ for $t>1/2$. On the other hand we see by Plots $4$ and $5$ that a jump induces a random shift, i.e. sBM$=$tBM for $t\leq 1/2$ and tBM$+W_{1/2}=$sBM for $t>1/2$.}
	\label{fig:pdiff}
\end{figure} 

\begin{figure}
	\begin{center}
	\includegraphics[scale=0.5]{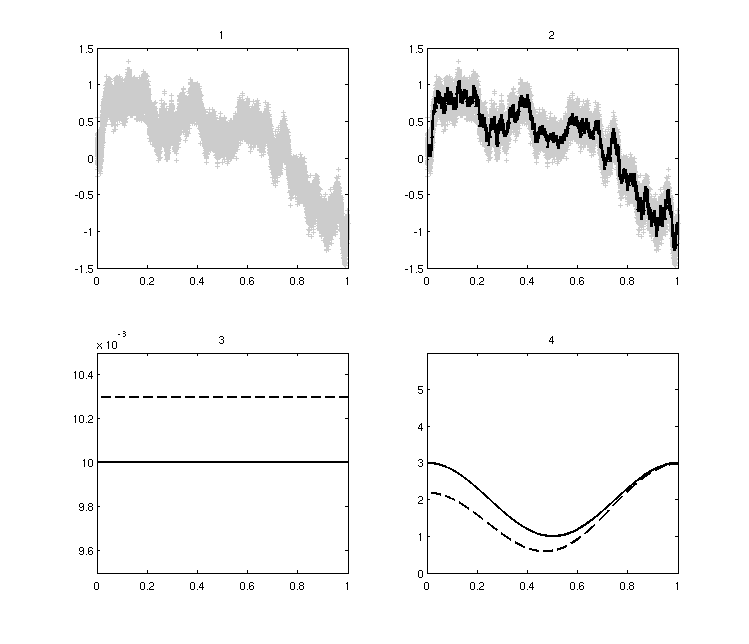}
	\end{center}
	\caption{$n=25000$ data points from model (\ref{eq.mod2}), $\epsilon_{i,n}\sim \mathcal{N}\left(0,1\right)$, i.i.d., $\tau=0.1$, $\sigma(t)=\left(2+\cos\left(2\pi t\right)\right)^{1/2}.$ Plot $1$ shows the data. Additionally to the data, we plotted the path of the tBM in Plot $2$. The reconstruction of $\tau^2$ and $\sigma^2$ (dashed lines) as well as the true function (solid lines) are given in Plot $3$ and $4$, respectively. The threshold parameters were selected as $N^*=1$ for estimation of $\tau^2$ and $N^*=3$ for estimation of $\sigma^2.$}
	\label{fig:s1p}
\end{figure} 

\begin{figure}
	\begin{center}
	\includegraphics[scale=0.5]{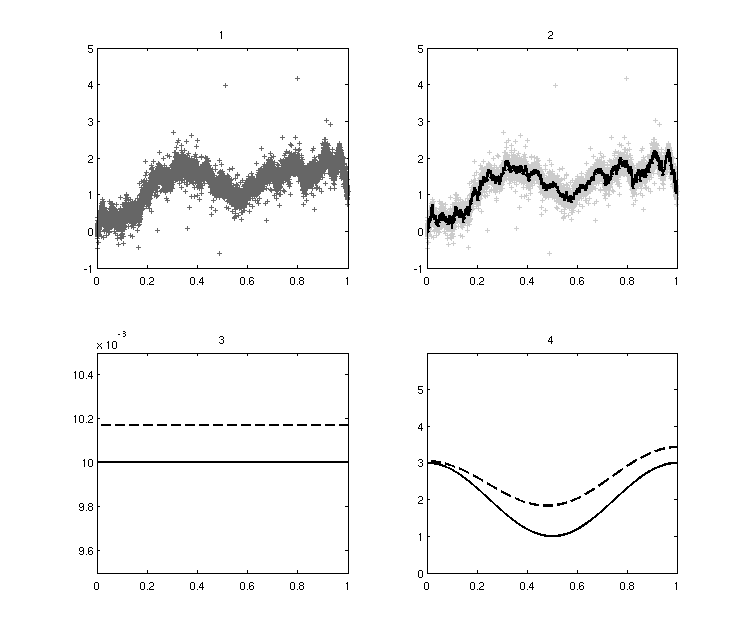}
	\end{center}
	\caption{(Heavy-tailed microstructure noise) \  As Figure \ref{fig:s1p} but instead of Gaussian errors we assumed that the noise follows a normalized Student's t-distribution with $3$ degrees of freedom. We observe that performance of $\hat \tau^2$  and $\hat \sigma^2$ is quite robust to heavy-tailed noise. The threshold parameters $N^*$ were selected as $1$ and $3$ for estimation of $\tau^2$ and $\sigma^2$, respectively.}
	\label{fig:s1t2p}
\end{figure}

\begin{figure}
	\begin{center}
	\includegraphics[scale=0.5]{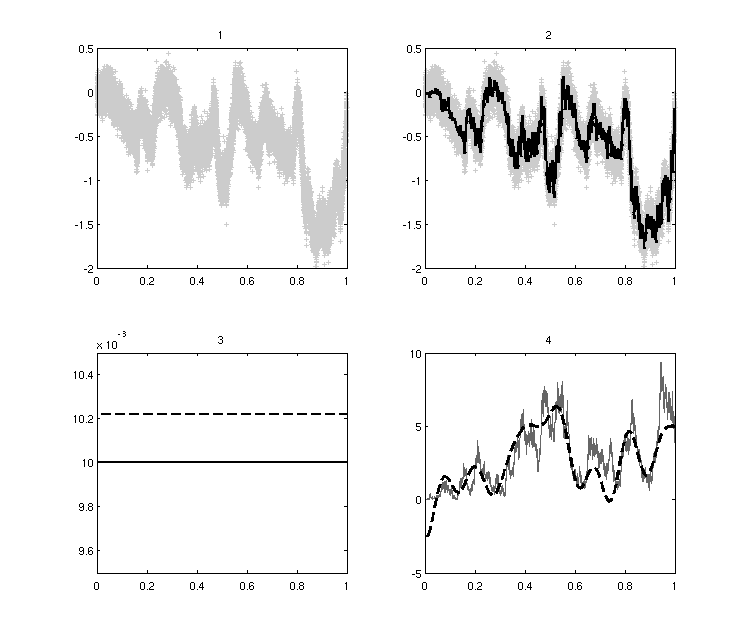}
	\end{center}
	\caption{(Low-smoothness) As Figure \ref{fig:s1p} but we chose $\sigma(t)=3\left|\tilde W_t\right|$, where $\left(\tilde W_t\right)_{t\in \left[0,1\right]}$ denotes a Brownian motion independent of the noise and the Brownian motion in (\ref{eq.mod2}). The estimator returns a smoothed version of the path. The threshold parameters $N^*$ were selected as $1$ and $17$ for estimation of $\tau^2$ and $\sigma^2$, respectively.}
	\label{fig:s3p}
\end{figure}

\begin{figure}
	\begin{center}
	\includegraphics[scale=0.5]{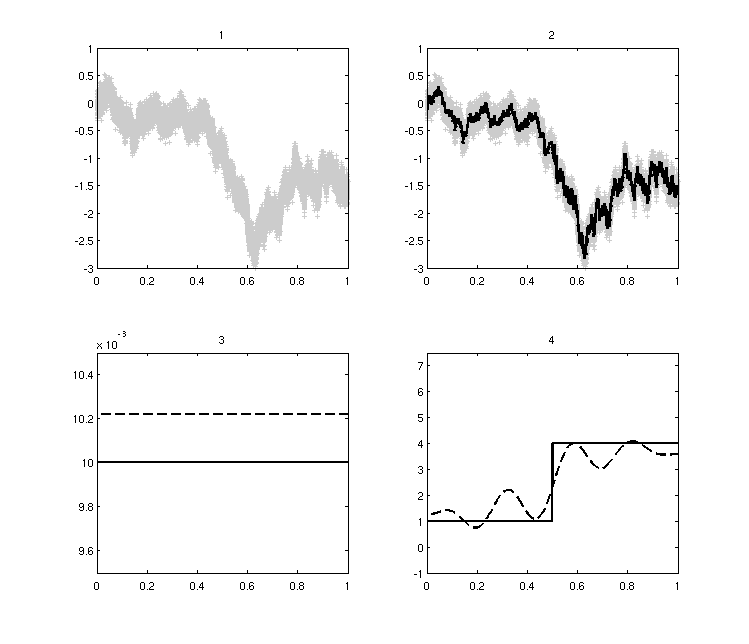}
	\end{center}
	\caption{(Jump function) As Figure \ref{fig:s1p} but we chose $\sigma(t)=1+\mathbb{I}_{\left(\right. 1/2,1\left. \right]}\left(t\right).$ The Gibbs phenomenon is clearly visible. The threshold parameters $N^*$ were selected as $1$ and $10$ for estimation of $\tau^2$ and $\sigma^2$, respectively.}
	\label{fig:s5p}
\end{figure}

\newevenside

{\LARGE \begin{center}
Supplementary Material: Nonparametric Estimation of the Volatility Function in a High-Frequency Model corrupted by Noise
\end{center}}

\medskip

{\large \begin{center}
Axel Munk$^*$ and Johannes Schmidt-Hieber
\end{center}}
\vspace{0.1cm}
{\begin{center}
{\em Institut f\"ur Mathematische Stochastik, Universit\"at G\"ottingen,} \\ {\em Goldschmidtstr. 7, 37077 G\"ottingen} \\ 
{\small {\em Email:} \texttt{munk@math.uni-goettingen.de}, \texttt{schmidth@math.uni-goettingen.de}}
\end{center}}

\bigskip



\begin{abstract}
This note provides proofs and supplementary technicalities for the paper "Nonparametric Estimation of the Volatility Function in a High-Frequency Model corrupted by Noise". In particular a proof on the rate of convergence of the estimator $\hat \tau^2_N$ is given.
\end{abstract}

\medskip

\noindent\textbf{AMS 2000 Subject Classification:}
Primary 62M09, 62M10; secondary 62G08.

\noindent\textbf{Keywords:\/} Brownian motion; Variance estimation; Minimax rate; Microstructure noise; Sobolev Embedding.

\section{Convergence Rate of $\hat \tau^2$}
\label{sec.proofs}
\subsection{Preliminary Results and Notation}

First we recall some notation. Let $\sigma_k(i/n):=\sigma(i/n)f_k(i/n)$ and $\tau_k(i/n):=\tau(i/n)f_k(i/n)$. Let throughout the following for the Sobolev {\it s}-ellipsoids in Definition \ref{def.ssobo} for $\sigma^2$ the constants being $l=\sigma_{\min}$ and $u=\sigma_{\max}$ and for $\tau^2$, $l=\tau_{\min}$, $u=\tau_{\max}$. We define $K_n:=n^{1/2}/\log n \ $ and
\begin{align*}
	\phi_{n}
	&:= \sup_{\sigma^2\in \Theta^b_s\left(\alpha,Q\right)}
	\max_{i=1,\ldots,n-1}\sup_{\xi \in \left[i/n,i+1/n\right]}
	\left|\sigma\left(\xi\right)-\sigma\left(\frac in\right)\right|,
	\notag \\
	\bar \phi_{n,1/2}
	&:=  \sup_{\tau^2\in \Theta^b_s\left(\beta,\bar Q\right)}
		\max_{k \leq K_n}
		\max_{i=1,\ldots,n}
	\sup_{\xi_i\in \left[\left(i-1\right)/n,i/n\right]}
	\left|\tau_k\left(\xi_i\right)-\tau_k\left(\frac in\right)\right|.
\end{align*}

\begin{prop2}
\label{prop.prop21}
Assume $\alpha,\beta>1/2$. It holds for any $\delta>0$,
\begin{align*}
	\phi_{n}&=O\left(n^{1/2-\alpha}+n^{\delta-1}\right) \\
	\bar\phi_n &= O\left(n^{1/2-\beta}+n^{-3/4}\right) \\
	\bar\phi_{n,1/2} &= O\left(n^{1/2-\beta}+n^{-1/2}\log^{-1}n\right).
\end{align*}
\end{prop2}

\begin{proof}
We only prove the third equality the other two can be deduced similarly. Note that for $\tau^2\in \Theta^b\left(\beta,Q\right)$,
\begin{align*}
	\left|\tau_k\left(\xi_i\right)-\tau_k\left(\frac in\right)\right|
	&\leq \sqrt{2}\left|\tau\left(\xi_i\right)-\tau\left(\frac in\right)\right|
	+\tau_{\max}^{1/2}k\pi/\left(2n\right) \\
	&\leq \frac1{\sqrt{2}\tau_{\min}}
	\left|\tau^2\left(\xi_i\right)-\tau^2\left(\frac in\right)\right|
	+\tau_{\max}^{1/2}k\pi/\left(2n\right).
\end{align*}
Taking supremum and applying Lemma \ref{lem.sobolem} gives the result.
\end{proof}

\subsection{Proofs for Estimation of $\tau^2$}

\begin{lem2}
\label{lem.tlem}
Let $\hat t_{k,0}$ be defined as in (\ref{eq.deftk}). Further assume $\alpha, \beta>1/2$, $Q, \bar Q>0$, $0<\sigma_{\min}\leq \sigma_{\max}<\infty,$ $0<\tau_{\min}\leq \tau_{\max}<\infty$ and $k=k_n \in \mathbb{N}$. Assume model (\ref{eq.mod2}). Then
\begin{align}
	\sup_{\sigma^2\in \Theta^b_s\left(\alpha,Q\right), \tau^2\in \Theta^b_s\left(\beta,\bar Q\right)}
	&\max_{k \leq K_n}\left|\E\hat t_{k,0}-t_{k,0}\right| 
	= 
	O\left(n^{1/2-\beta}\log^{1/2} n\right)+o\left(n^{-1/2}\right),
	\label{eq.Et2}\\
	\sup_{\sigma^2\in \Theta^b_s\left(\alpha,Q\right), \ \tau^2\in \Theta^b_s\left(\beta,\bar Q\right)}
	&\max_{k \leq K_n} \Var\left(\hat t_{k,0}\right)
	=
	O\left(n^{-1}\right). \label{eq.Vt2}
\end{align}
If further $\epsilon$ is $n$-variate standard normal, then
\begin{align}
	\sup_{\sigma^2\in \Theta^b_s\left(\alpha,Q\right), \ \tau^2\in \Theta^b_s\left(\beta,\bar Q\right)}
	&\max_{k \leq K_n}\left|\Var\left(\hat t_{k,0}\right)
	-\frac 2n \int_0^1\tau_k^4(x) dx\right| \notag \\
	&=
	O\left(n^{-1}\log^{-1} n\right), \label{eq.Vt2normal}\\
	n^{1/2}\left(\hat t_{k,0}-t_{k,0}\right)
	&\stackrel{L}{\rightarrow} \mathcal{N}\left(0,2\int_0^1\tau_k^4(x)dx \right)
	\quad \text{for} \ \beta>1, \ k\leq K_n.
	\label{eq.ant2}
\end{align}
Assume model (\ref{eq.mod}). Then it holds
\begin{align}
	&\sup_{\sigma^2\in \Theta^b_s\left(\alpha,Q\right), \ \tau^2\in \Theta^b_s\left(\beta,\bar Q\right)}
	\max_{k \leq K_n}\left|\E\hat t_{k,0}-t_{k,0}\right| \notag \\
	& \quad\quad\quad\quad \quad = 
	O\left(n^{1/2-\beta}\log^{1/2}n+n^{1-2\alpha}\log ^2n\right)+o\left(n^{-1/2}\right), \label{eq.Et}\\
	&\sup_{\sigma^2\in \Theta^b_s\left(\alpha,Q\right), \ \tau^2\in \Theta^b_s\left(\beta,\bar Q\right)}
	\max_{k \leq K_n}\Var\left(\hat t_{k,0}\right)
	=
	O\left(n^{2-4\alpha}\log^4 n+n^{-1}\right).
	\label{eq.Vt}
\end{align}
If further $\epsilon$ is $n$-variate standard normal, then
\begin{align}
	\sup_{\sigma^2\in \Theta^b_s\left(\alpha,Q\right), \tau^2\in \Theta^b_s\left(\beta,\bar Q\right)}
	&\max_{k \leq K_n}\left|\Var\left(\hat t_{k,0}\right)
	-\frac 2n \int_0^1\tau_k^4(x) dx \right| \notag \\ 
	&\quad\quad\quad =
	O\left(n^{2-4\alpha}\log^4 n+n^{-1}\log^{-1} n\right),
	\label{eq.Vtnormal} \\
	n^{1/2}\left(\hat t_{k,0}-t_{k,0}\right)
	&\stackrel{L}{\rightarrow} \mathcal{N}\left(0,2\int_0^1\tau_k^4(x)dx \right)
	\quad \text{for} \ \alpha>3/4, \ \beta>1, \ k\leq K_n.
	\label{eq.ant}
\end{align}
\end{lem2}

\begin{proof}
Again we work with the generalized estimators as introduced in Section \ref{sec.preandnot}. As in the proof of Lemma \ref{lem.sklem} we introduce for two centered random vectors $P$ and $Q$ a semi-inner product defined by $\left\langle P, Q \right\rangle_\tau :=  \E\left(P^tDJ^{\tau}_nD^tQ\right)$  and obtain
\begin{align}
	\E\hat t_{k,0,l}
	&=
	\left\langle X_{1,k}, X_{1,k} \right\rangle_\tau 
	+ \left\langle X_{2,k}, X_{2,k} \right\rangle_\tau 
	+ \left\langle Z_{1,k,l}, Z_{1,k,l} \right\rangle_\tau 
	+ \left\langle Z_{2,k}, Z_{2,k} \right\rangle_\tau \notag \\
	& \quad
	+2 \left\langle X_{1,k}, Z_{1,k,l} \right\rangle_\tau 
	+2 \left\langle X_{2,k}, Z_{2,k} \right\rangle_\tau.
	\label{eq.eex}
\end{align}
First we bound $\left \langle X_{2,k}, X_{2,k} \right\rangle_\tau$, which will turn out to be the leading term. Similar to (\ref{eq.x2x2}) we have
\begin{eqnarray*}
	\left\langle X_{2,k}, X_{2,k} \right\rangle_\tau
	= \Tr\left(\Lambda  J^\tau_n D^t T_k^2 D\right)
	+ \frac 12 \Tr\left(J^\tau_n D^t \tilde T_k D\right),
\end{eqnarray*}
and due to
\begin{eqnarray}
	\Tr\left(J^\tau_n \right)=O\left(\log n \right)
	\label{eq.tr2}
\end{eqnarray}
the same argument as for (\ref{eq.t2}) gives
\begin{eqnarray*}
	\sup_{\tau^2\in \Theta_s^b\left(\beta, \bar Q\right)}
	\max_{k \leq K_n}\Tr\left(J_n^\tau D\tilde T_k D\right)
	=O\left(\bar \phi_{n,1/2}^2 \log n \right).
\end{eqnarray*}
Hence this is a negligible term. Using Lemma \ref{lem.approx} (iii)
\begin{align*}
	\Tr\left(\Lambda  J^\tau_n D^t T_k^2 D\right) 
	&= \left(n-n/\log n \right)^{-1}
	\sum_{i=\left[n/\log n\right]}^{n-1}
	\left(A\left(\tau^2_k,0\right)-
	A\left(\tau^2_k,2i\right)\right) \\
	&= \bar r_n A\left(\tau^2_k,0\right)
	-\left(n-n/\log n \right)^{-1}\sum_{i=\left[n/\log n\right]}^{n-1}
	A\left(\tau^2_k,2i\right),
\end{align*}
where $\bar r_n=\left(n-\left[n/\log n\right]\right)/\left(n-n/\log n \right).$ Note $1-\bar r_n \leq 1/\left(n-n/\log n \right).$ By Lemma \ref{lem.bound}
\begin{align*}
	& \sup_{\tau^2\in \Theta_s^b\left(\beta, \bar Q\right)} \max_{k \leq K_n}
	\left| \Tr\left(\Lambda  J^\tau_n D^t T_k^2 D\right) -t_{k,0}\right| \\
	& \quad \leq 
	\sup_{\tau^2\in \Theta_s^b\left(\beta, \bar Q\right)} \max_{k \leq K_n}
	\left(
	\left(1-\bar r_n\right)\left|t_{k,0}\right|+\sum_{m=n}^\infty \left|t_{k,m}\right|
	+ 2\left(n-n/\log n \right)^{-1} \sum_{i=0}^\infty \left|t_{k,i}\right|\right) \\
	& \quad \leq 2C_{\beta,Q}n^{1/2-\beta}+6\left(n-n/\log n \right)^{-1}
	\sup_{\tau^2\in \Theta_s^b\left(\beta, \bar Q\right)}
	\sum_{i=0}^\infty \left|t_{0,i}\right|= O\left(n^{-1}+n^{1/2-\beta}\right).
\end{align*}
This shows that 
\begin{align*}
	\sup_{\tau^2\in \Theta_s^b\left(\beta, \bar Q\right)}
	\max_{k \leq K_n}\left|\left\langle X_{2,k}, X_{2,k} \right\rangle_\tau-t_{0,k}\right| = O\left(n^{-1}+\bar \phi_{n,1/2}^2 \log n  + n^{1/2-\beta}\right).
\end{align*}
The remaining part of the proofs of (\ref{eq.Et}) and (\ref{eq.Et2}) is concerned with bounding the other expressions in (\ref{eq.eex}). Applying Lemma \ref{lem.p}, we obtain 
\begin{align*}
	\left\langle X_{1,k}, X_{1,k} \right\rangle_\tau 
	= 
	\frac 1n \Tr\left(D
	J^\tau_n D^t
	\Sigma_k^2
	\right)
	\leq 2\sigma_{\max}
	\frac 1n \Tr\left(J_n^\tau\right)
\end{align*}
implying
\begin{align*}
	\sup_{\sigma^2\in \Theta_s^b\left(\alpha, Q\right)}
	\max_{k \leq K_n}
	\left\langle X_{1,k}, X_{1,k} \right\rangle_\tau =O\left(n^{-1}\log n \right).
\end{align*}
We obtain with Lemma \ref{lem.eigbound} in the same way as in (\ref{eq.zik3zik3}), (\ref{eq.zik2zik2}) and (\ref{eq.z2klz2kl})
\begin{align*}
	&\sup_{\sigma^2\in \Theta_s^b \left(\alpha, Q\right)}
	\max_{k \leq K_n} \ 
	\left\langle Z_{1,k,1}, Z_{1,k,1} \right\rangle_\tau
	=O\left(n^{-1}\log n \ \phi_n^2\right), \\
	&\sup_{\sigma^2\in \Theta_s^b \left(\alpha, Q\right)}
	\max_{k \leq K_n}
	\left\langle Z_{1,k,2}, Z_{1,k,2} \right\rangle_\tau
	= O\left(\log^2 n \ \phi_{n}^2\right), \\
	&\sup_{\tau^2\in \Theta_s^b \left(\beta, \bar Q\right)}
	\max_{k \leq K_n}
	\left\langle Z_{2,k}, Z_{2,k} \right\rangle_\tau=O\left(\bar \phi_{n,1/2}^2 \log n\right).
\end{align*}
From the Cauchy-Schwarz inequality follows that 
\begin{align*}
	\left | \left\langle X_{1,k}, Z_{1,k,l} \right\rangle_\tau \right|
	&\leq \left\langle X_{1,k}, X_{1,k} \right\rangle_\tau ^{1/2}
	\left\langle Z_{1,k,l}, Z_{1,k,l} \right\rangle_\tau ^{1/2} \\
	&\leq
	\left\langle X_{1,k}, X_{1,k} \right\rangle_\tau
	+
	\left\langle Z_{1,k,l}, Z_{1,k,l} \right\rangle_\tau,
\end{align*}
\begin{align*}
	\left | \left\langle X_{2,k},Z_{2,k} \right\rangle_\tau \right|
	&\leq \left\langle X_{2,k}, X_{2,k} \right\rangle_\tau ^{1/2}
	\left\langle Z_{2,k}, Z_{2,k} \right\rangle_\tau ^{1/2}
	.
\end{align*}
This yields
\begin{align*}
	& \sup_{\sigma^2\in \Theta^b_s\left(\alpha,Q\right), \ \tau^2\in \Theta^b_s\left(\beta,\bar Q\right)}
	\max_{k \leq K_n}
	\left|E\hat t_{k,0}-t_{k,0}\right| \\
	& \quad \quad =
	\begin{cases}
 	O\left(n^{-1}\log n  +n^{1/2-\beta}
	+
	\bar \phi_{n,1/2} \log^{1/2} n\right)\quad &\text{for} \ l=1, \\
	O\left(
	n^{-1}\log n +n^{1/2-\beta}
	+
	\bar \phi_{n,1/2} \log^{1/2} n
	+
	\phi_{n}^2\log^2 n
	\right)
	 \quad &\text{for} \ l=2,
	\end{cases}
\end{align*}
and therefore (\ref{eq.Et}) and (\ref{eq.Et2}) holds by Proposition \ref{prop.prop21}. In order to calculate the covariance we use the decomposition (\ref{eq.deltayk}).
We have
\begin{eqnarray*}
	\hat t_{k,0,l}
	= \xi^t C_{1,k,l}^t DJ_n^\tau D 
		C_{1,k,l}\xi
	+ 2 \xi^t C_{1,k,l}^t DJ_n^\tau D C_{2,k} \epsilon
	+ \epsilon^t C_{2,k}^t DJ_n^\tau D C_{2,k} \epsilon.
\end{eqnarray*}
Using the CS-inequality repeatedly, we can write
\begin{align}
	&\left|\Var\left(\hat t_{k,0,l}\right)
	- \Var\left(\epsilon^t C_{2,k}^t DJ_n^\tau D C_{2,k} \epsilon\right)\right|
	\label{eq.var} \\ 
	&\quad\quad
	\leq \left( \Var^{1/2}\left(\xi^t C_{1,k,l}^t DJ_n^\tau D C_{1,k,l}\xi \right)	
	+
	2\Var^{1/2}\left( 
	\xi^t C_{1,k,l}^t DJ_n^\tau D C_{2,k} \epsilon
	\right)\right)^2 \notag \\ 
	& \quad\quad\quad \ 
	+ \left( \Var^{1/2}\left(\xi^t C_{1,k,l}^t DJ_n^\tau D C_{1,k,l} \xi \right)	
	+
	2\Var^{1/2}\left( 
	\xi^t C_{1,k,l}^t DJ_n^\tau D C_{2,k} \epsilon
	\right)\right) \notag \\
	& \quad\quad\quad \    \cdot 2\Var^{1/2}\left(\epsilon^t C_{2,k}^t DJ_n^\tau D C_{2,k} \epsilon\right) \notag
\end{align}
We subdivide the remaining part of the proofs of (\ref{eq.Vt}) and (\ref{eq.Vt2}) into three steps (a), (b) and (c), where we calculate  $\Var\left(\epsilon^t C_{2,k}^t DJ_n^\tau D C_{2,k} \epsilon\right)$, $\Var\left(\xi^t C_{1,k,l}^t DJ_n^\tau D C_{1,k,l}\xi \right)$ and $\Var\left( \xi^t C_{1,k,l}^t DJ_n^\tau D C_{2,k} \epsilon\right)$, respectively.

\medskip

\noindent
{\em (a)} Let $\TrSq(A):=\sum_{i=1}^n A_{i,i}^2$ for $A\in \mathbb{M}_n.$ Then by Lemma \ref{lem.n} it follows
\begin{align*}
	&\Var\left(\epsilon^t C_{2,k}^t DJ_n^\tau D C_{2,k} \epsilon\right) \\ &\quad\quad=2\left\|C_{2,k}^tDJ_n^{\tau}DC_{2,k}\right\|_F^2
	+\Cum_4\left(\epsilon\right)\TrSq\left(C_{2,k}^tDJ_n^{\tau}DC_{2,k}\right) \\
	&\quad\quad\leq \left(2+\Cum_4\left(\epsilon\right)\right)
	\left\|\left(J_n^\tau\right)^{1/2}D\Cov\left(X_{2,k}+Z_{2,k}\right)
	D\left(J_n^\tau\right)^{1/2}\right\|_F^2,
\end{align*}
where equality holds if $\Cum_4\left(\epsilon\right)=0$. By Proposition \ref{prop.eqprop} we see that 
\begin{align*}
	\sup_{\tau^2\in \Theta_s^b\left(\beta, \bar Q\right)}\max_{k \leq K_n}
	&\left|\Var\left(\epsilon^t C_{2,k}^t DJ_n^\tau D C_{2,k} \epsilon\right)
	-\frac2n \int_0^1 \tau_k^4(x)dx\right| \\
	& = O\left(\Cum_4\left(\epsilon\right)n^{-1}+n^{-1}\log^{-1}n\right).
\end{align*}

\bigskip

\noindent
(b) In this part of the proof we will bound $\Var\left(\xi^t C_{1,k,l}^t DJ_n^\tau D C_{1,k,l}\xi \right)$. Similar to part (a) in the proof of Lemma \ref{lem.sklem} it holds
\begin{align*}
	\Var\left(\xi^tC_{1,k,l}^tDJ_n^\tau DC_{1,k,l}\xi\right) 
	&\leq 16\lambda_1^2\left(J_n^\tau\right)
 	\left(
 	\left\|\Cov\left(X_{1,k}\right)\right\|_F^2
	+\left\|\Cov\left(Z_{1,k,l}\right)\right\|_F^2\right) 
	\\
	&\leq 
	\begin{cases}
	2^8n^{-2}\log^4 n\left(n^{-1}\sigma_{\max}^2+4n^{-1}\phi_{n}^4\right),
	\quad &l=1, \\
	2^8n^{-2}\log^4 n\left(n^{-1}\sigma_{\max}^2+4n^{2}\phi_n^4\right),
	\quad & l=2,
	\end{cases}
\end{align*}
where we used Lemma \ref{lem.eigbound} in the second inequality. Hence we get
\begin{align}
	\sup_{\sigma^2\in \Theta_s^b\left(\alpha,Q\right)}
	\max_{k \leq K_n}
	&\Var\left(\xi^tC_{1,k,l}^tDJ_n^\tau DC_{1,k,l}\xi\right)
	= 
	\begin{cases}
	O\left(n^{-3}\log^4 n\right), &l=1, \\
	O\left(
	\log^4 n
	\left(\phi_{n}^4 +n^{-3}\right)
	\right),  &l=2.
	\end{cases}
	\label{eq.var11}
\end{align}

\medskip

\noindent
(c) Using Lemma \ref{lem.n} (ii)
\begin{align*}
	&\Var\left(\xi^tC_{1,k,l}^tDJ_n^\tau DC_{2,k}\epsilon\right) \leq
	\frac 1{\sqrt{2}}\Var^{1/2}\left(
	\xi^tC_{1,k,l}^tDJ_n^\tau DC_{1,k,l}\xi
	\right)
	\left\|C_{2,k}^tDJ_n^\tau DC_{2,k}\right\|_F
\end{align*}
and hence
\begin{align*}
	\sup_{\sigma^2\in \Theta^b_s\left(\alpha,Q\right), \ \tau^2\in \Theta^b_s\left(\beta,\bar Q\right)}
	&\max_{k \leq K_n}
	\Var\left(\xi^tC_{1,k,l}^tDJ_n^\tau DC_{2,k}\epsilon\right)
	\\
	& = 
	\begin{cases}
	O\left(n^{-2}\log^2 n\right), \ & l=1, \\
	O\left(\log^2 n
	\left(\phi_{n}^2+n^{-3/2}\right)\right)
	O\left(n^{-1/2}\right), \ & l=2.
	\end{cases}
\end{align*}
Combining (a), (b) and (c) in (\ref{eq.var}) yields
\begin{align}
	&\sup_{\sigma^2\in \Theta^b_s\left(\alpha,Q\right), \ \tau^2\in \Theta^b_s\left(\beta,\bar Q\right)}
	\max_{k \leq K_n}
	\left|\Var\left(\hat t_{k,0}\right)
	- \frac 2n \int_0^1 \tau_k^4(x)
	\right| \notag \\
	& \ \ = O\left(\frac{\Cum_4\left(\epsilon\right)}n+\frac 1{n\log n}\right)+
	\begin{cases}
	O\left(\phi_{n}^4\log^4 n
	+\phi_{n}n^{-3/4}\log n 
	\right), \  & l=1, \\
	0, \  & l=2,
	\end{cases}
	\label{eq.varall}
\end{align}
and hence (\ref{eq.Vt2}), (\ref{eq.Vt2normal}), (\ref{eq.Vt}) and (\ref{eq.Vtnormal}) follow using Proposition \ref{prop.prop21}.

\medskip

Finally we will show the asymptotic normality (\ref{eq.ant}) and (\ref{eq.ant2}). Because of the decomposition (\ref{eq.deltayk}), we have
\begin{eqnarray*}
	\hat t_{k,0,l}
	=\xi^tC_{1,k,l}^tDJ_n^\tau DC_{1,k,l}\xi
	+2\xi^tC_{1,k,l}^tDJ_n^\tau DC_{2,k}\epsilon
	+\epsilon^tC_{2,k}^tDJ_n^\tau C_{2,k}\epsilon.
\end{eqnarray*}
As proved above $n^{1/2}\left(\xi^tC_{1,k,l}^tDJ_n^\tau DC_{1,k,l}\xi
	+2\xi^tC_{1,k,l}^tDJ_n^\tau DC_{2,k}\epsilon\right)\stackrel{P}{\rightarrow} 0$ for $\beta>1, \alpha>3/4$ if $l=1$ and $\beta>1$ if $l=2$. Hence by Slutzky's Lemma it suffices to show that
\begin{eqnarray*}
	n^{1/2}\left(\epsilon^tC_ {2,k}^t DJ_n^\tau DC_{2,k}\epsilon
	-\E\left(\epsilon^tC_ {2,k}^t DJ_n^\tau DC_{2,k}\epsilon\right)\right)
	\stackrel{L}{\rightarrow} \mathcal{N}\left(0,2\int_0^1\tau_k^4(x)dx\right).
\end{eqnarray*}
In order to apply Theorem \ref{thm.wclt}, it remains to show
\begin{eqnarray*}
	n^{1/2}\lambda_1\left(C_{2,k}^tDJ_n^\tau DC_{2,k}\right)
	\rightarrow 0.
\end{eqnarray*}
Using Corollary \ref{cor.eig}, we see that
\begin{align*}
	n^{1/2}&\lambda_1\left(C_{2,k}^tDJ_n^\tau DC_{2,k}\right)
	\leq 4n^{-1/2}\log^2 n\lambda_1\left(\Cov\left(X_{2,k}+Z_{2,k}\right)\right) \\
	&\quad \leq 
	8n^{-1/2}\log ^2 n \ \lambda_1\left(\Cov\left(X_{2,k}\right)\right)
	+8n^{-1/2}\log^2 n\lambda_1\left(\Cov\left(Z_{2,k}\right)\right)\\
	&\quad \leq 8n^{-1/2}\log^2 n\sup_{t\in\left[0,1\right]}\tau_k^2(t)\max_i\lambda_i
	+
	8n^{-1/2}\log^2 n\phi_{n,1/2}^2
	=o\left(1\right),
\end{align*}
which yields the last statement of the lemma.
\end{proof}

\begin{proof}[Proof of Theorem \ref{thm.misetthm}]
The proof is close to the one of Theorem \ref{thm.misesthm}. We obtain
\begin{align*}
	&\sup_{\sigma^2\in \Theta^b_s\left(\alpha,Q\right), \ \tau^2\in \Theta^b_s\left(\beta,\bar Q\right)}
	\int_0^1\Bias\left(\hat \tau_N^2(t)\right)dt \\
	&\quad =
	O\left(N \sup_{\sigma^2\in \Theta^b_s\left(\alpha,Q\right), \ \tau^2\in \Theta^b_s\left(\beta,\bar Q \right)}
	\max_{0\leq k\leq N}
	\left|\E\left(\hat t_{k,0,l}\right)
	-t_{k,0}\right|^2+N^{-2\beta}\right), \\
	&\sup_{\sigma^2\in \Theta^b_s\left(\alpha,Q \right), \ \tau^2\in \Theta^b_s\left(\beta,\bar Q \right)}
	\int_0^1\Var\left(\hat \tau_N^2(t)\right)dt  \\
	&=
	O\left(N
	\sup_{\sigma^2\in \Theta^b_s\left(\alpha,Q \right), \ \tau^2\in \Theta^b_s\left(\beta,\bar Q \right)}
	\max_{0\leq k\leq N}\Var\left(
	\hat t_{k,0,l}\right)\right).
\end{align*}
\end{proof}

\section{Technical Results}
\label{sec.apptechres}
\begin{prop2}
\label{prop.X1R1prop}
Let $A\in \mathbb{M}_{n-1}$. Then
\begin{align*}
	\Tr\left(J_nDAD\right)
	\leq \left(n+5n^{3/2}+8n^{3/2}\left(1+\log n \right)\right) \max_{i,j}\left|\left(A\right)_{i,j}\right|.
\end{align*}
\end{prop2}

\begin{proof}
Write $A=\left(a_{i,j}\right)_{i,j=1,\ldots,n-1}.$ Note that
\begin{eqnarray*}
	\left(DAD\right)_{i,j}
	=
	\frac 2{n}
	\sum_{p,q=1}^{n-1}
	\sin\left(\frac{ip\pi}n\right)\sin\left(\frac{qj\pi}n\right)
	a_{p,q}.
\end{eqnarray*}
For $i=j$ we have further
\begin{align}
	\left(DAD\right)_{i,i}
	&=
	\frac 1{n}
	\sum_{p,q=1}^{n-1}
	a_{p,q} \cos\left(\left(p-q\right)\frac {i\pi}n\right) 
	+
	\frac 1{n}
	\sum_{p,q=1}^{n-1}
	a_{p,q} \cos\left(\left(p+q\right)\frac {i\pi}n\right).
	\label{eq.traceexplicit}
\end{align}
In order to bound the r.h.s. we need bounds for 
\begin{align*}
	&\left|\sum_{i=\left[n^{1/2}\right]+1}^{2\left[n^{1/2}\right]}
	\cos\left(r\frac{i\pi}n\right) \right|
	\leq
	\left| \Dir_{2\left[n^{1/2}\right]}\left(r\pi/n\right)
	- \Dir_{\left[n^{1/2}\right]}\left(r\pi/n\right)\right| \\
	&\quad \quad\leq
	\left|
	\frac {\sin\left(\left(2\left[n^{1/2}\right]+1/2\right)r\pi/n\right)}
	{2\sin\left(r\pi/\left(2n\right)\right)}
	\right|
	+
	\left|
	\frac {\sin\left(\left(\left[n^{1/2}\right]+1/2\right)r\pi/n\right)}
	{2\sin\left(r\pi/\left(2n\right)\right)}
	\right| \\
	&\quad\quad \leq \left|
	\frac 1
	{\sin\left(r\pi/\left(2n\right)\right)}
	\right|.
\end{align*}
Let $B_1:=\left\{1,
\ldots,n\right\}$ and $B_2:= \left\{n+1,\ldots,2n-2\right\}$. Then 
\begin{align*}
	\left|\sum_{i=\left[n^{1/2}\right]+1}^{2\left[n^{1/2}\right]}
	\cos\left(r\frac{i\pi}n\right) \right|&\leq
	\begin{cases}
	 n^{1/2} \ \  &\text{for} \ \   r=0, \\
	4n/r \ \ &\text{for} \ \  r\in B_1, \\
	n/\left(2n-r\right) \ \ &\text{for} \ \ r\in B_2.
	\end{cases}
\end{align*}
Therefore, we can bound the first term of the r.h.s. of (\ref{eq.traceexplicit}) by
\begin{align*}
	&\left|\sum_{i=\left[n^{1/2}\right]+1}^{2\left[n^{1/2}\right]} \frac 1{n}
	\sum_{p,q=1}^{n-1}
	a_{p,q} \cos\left(\left(p-q\right)\frac {i\pi}n\right) \right| 
	\\
	&\quad\quad\leq
	\frac 1{n}\sum_{p,q=1}^{n-1}
	\left|a_{p,q}\right|
	\left|\sum_{i=\left[n^{1/2}\right]+1}^{2\left[n^{1/2}\right]} \cos\left(\left(p-q\right)\frac{i\pi}n\right)\right| \\
	&\quad\quad\leq
	n^{-1}\max_{p,q=1,\ldots,n-1} \left|a_{p,q}\right|
	\left(n^{3/2}+2\sum_{\underset{q-p\in B_1}{p,q=1}}^{n-1}
	\frac{4n}{q-p}\right)
\end{align*}
and the second term by
\begin{eqnarray*}
	&&\left|\sum_{i=\left[n^{1/2}\right]+1}^{2\left[n^{1/2}\right]} \frac 1{n}
	\sum_{p,q=1}^{n-1}
	a_{p,q} \cos\left(\left(p+q\right)\frac {i\pi}n\right) \right| \\
	&& \quad \leq
	n^{-1}\max_{p,q=1,\ldots,n-1} \left|a_{p,q}\right|
	\left(
	\sum_{\underset{p+q\in B_1}{p,q=1}}^{n-1}
	\frac{4n}{p+q}
	+ \sum_{\underset{p+q\in B_2}{p,q=1}}^{n-1}
	\frac{n}{2n-\left(p+q\right)}
	\right)
 	 \\
 	&& \quad
	\leq 5n\max_{p,q=1,\ldots,n-1} \left|a_{p,q}\right|.
\end{eqnarray*}
Due to
\begin{eqnarray*}
	\sum_{\underset{q-p\in B_1}{p,q=1}}^{n-1}
	\frac{1}{q-p}
	\leq n\sum_{r=1}^n \frac 1r \leq n\left(1+\log n \right)
\end{eqnarray*}
and
\begin{eqnarray*}
	\Tr\left(J_nDAD\right)
	&=& \sqrt{n}\sum_{i=\left[n^{1/2}\right]+1}^{2\left[n^{1/2}\right]} \left(DAD\right)_{i,i} \\
	&\leq& \left(n+5n^{3/2}+8n^{3/2}\left(1+\log n \right)\right)
	\max_{p,q=1,\ldots,n-1} \left|a_{p,q}\right|
\end{eqnarray*}
we obtain the result.
\end{proof}

\begin{prop2}
\label{prop.eqprop}
It holds
\begin{align*}
	\sup_{\tau^2\in \Theta_s\left(\beta,Q\right)} \max_{k\leq n^{1/2}}
	&\left|
	\left\|\left(J_n^\tau\right)^{1/2}D\Cov\left(X_{2,k}+Z_{2,k}\right)D\left(J_n^\tau\right)^{1/2}\right\|_F^2
	- 
	\frac 2n \int_0^1\tau_k^4(x)dx\right| \\
	&=O\left(n^{-1}\log^{-1} n\right).
\end{align*}
\end{prop2}

\begin{proof}
We obtain with (\ref{eq.tktdec}) $\Cov\left(X_{2,k}+Z_{2,k}\right)=1/2 T_k^2K+1/2KT_k^2+S_{k}$, where $S_{k}:= 1/2\tilde T_{k}+ \Cov\left(X_{2,k},Z_{2,k}\right)+\Cov\left(Z_{2,k},X_{2,k}\right)+\Cov\left(Z_{2,k}\right).$ Application of the triangle inequality gives
\begin{align}
	&\frac 12\left\|\left(J_n^\tau\right)^{1/2}D\left(T_k^2K+KT_k^2\right)D\left(J_n^\tau\right)^{1/2}\right\|_F 
	-
	\left\|\left(J_n^\tau\right)^{1/2}DS_{k}D\left(J_n^\tau\right)^{1/2}\right\|_F \notag \\
	& \quad\leq \left\|\left(J_n^\tau\right)^{1/2}D\Cov\left(X_{2,k}+Z_{2,k}\right)D\left(J_n^\tau\right)^{1/2}\right\|_F \notag \\
	& \quad \leq
	\frac 12\left\|\left(J_n^\tau\right)^{1/2}D\left(T_k^2K+KT_k^2\right)D\left(J_n^\tau\right)^{1/2}\right\|_F 
	+
	\left\|\left(J_n^\tau\right)^{1/2}DS_{k}D\left(J_n^\tau\right)^{1/2}\right\|_F.
	\label{eq.normineqs}
\end{align}
Note that because of Lemma \ref{lem.froblem} (iii) it holds
\begin{align}
	\Tr\left(\left(\left(J_n^\tau\right)^{1/2}DT_k^2KD\left(J_n^\tau\right)^{1/2}\right)^2 \right) &\leq
	\frac 14\left\|\left(J_n^\tau\right)^{1/2}D\left(T_k^2K+KT_k^2\right)D\left(J_n^\tau\right)^{1/2}\right\|_F^2 \notag \\
	&\leq \left\|\left(J_n^\tau\right)^{1/2}DT_k^2KD
	\left(J_n^\tau\right)^{1/2}\right\|_F^2.
	\label{eq.frobineq}
\end{align}
Now we will bound 
\begin{align*}
	\Tr\left(\left(\left(J_n^\tau\right)^{1/2}DT_k^2KD\left(J_n^\tau\right)^{1/2}\right)^2 \right)
	&= \Tr\left(\left[\left(J_n^\tau \Lambda\right)^{1/2}DT_k^2D\left(\Lambda J_n^\tau\right)^{1/2}\right]^2\right) \\
	&= \sum_{i=1}^{n-1} \lambda_i^2\left(DT_k^2D\Lambda J_n^\tau\right)
\end{align*}
from below. We obtain with Lemma \ref{lem.p}
\begin{align*}
	&\lambda_i\left(DT_k^2D\Lambda J_n^\tau\right) \geq
	\begin{cases}
	\lambda_{n-\left[n/\log n\right]}\left(\Lambda J_n^\tau\right)
	\lambda_{\left[n/\log n\right] -1 +i}\left(DT_k^2D\right),
	&i\leq n-\left[n/\log n\right],
	\\
	0,  &i> n-\left[n/\log n\right].
	\end{cases}
\end{align*}
Denote by $\tau_{k,\left(i\right)}$ the $i$-th largest component of the vector $$\left(\tau_k\left(1/n\right),\ldots,\tau_k\left(1-1/n\right)\right).$$ Then
\begin{align}
	\Tr\left(\left(\left(J_n^\tau\right)^{1/2}\right.\right. & \left.\left. DT_k^2KD\left(J_n^\tau\right)^{1/2}\right)^2 \right)
	=
	\sum_{i=1}^{n-1}\lambda_i^2\left(DT_k^2D\Lambda J_n^\tau\right) \notag \\
	&\geq \sum_{i=1}^{n-\left[n/\log n\right]}
	\left(n-n/\log n \right)^{-2}
	\tau_{k,\left(\left[n/\log n\right] -1 +i\right)}^4 \notag \\
	&\geq \left(n-n/\log n \right)^{-2}
	\sum_{i=1}^{n-1} \tau_k^4 \left(\frac in\right) 
	- \tau_{\max}^2 \frac n{\log n}\left(n-n/\log n \right)^{-2}.
	\label{eq.lbtau}
\end{align}
Next we will derive an upper bound for the r.h.s. of (\ref{eq.frobineq}). Let analogously to the Definition (\ref{eq.ttildedef}) $\bar T_{k}$ be a tridiagonal matrix with entries
\begin{eqnarray*}
	\left(\bar T_{k}\right)_{i,j}:=
	\begin{cases}
	\left(\Delta_i \tau_k^2\right)^2 
	\quad &\text{for} \quad i= j-1, \\
	\left(\Delta_{j} \tau_k^2\right)^2 
	\quad &\text{for} \quad i= j+1, \\
	0 \quad &\text{otherwise}.
	\end{cases}
\end{eqnarray*}
Note that $\max_i \left|\Delta_i \tau_k^2\right|\leq 2\tau_{\max}^{1/2}\bar \phi_{n,1/2}.$ It is easy to check that $T_k^2KT_k^2=1/2T_k^4K+1/2KT_k^4+1/2\bar T_{k}$ holds. Clearly, $J_n^\tau\leq \left(n-n/\log n \right)^{-1}\Lambda^{-1}$, and therefore we have for the upper bound in (\ref{eq.frobineq})
\begin{align*}
	&\left\|\left(J_n^\tau\right)^{1/2}DT_k^2KD
	\left(J_n^\tau\right)^{1/2}\right\|_F^2
	\leq 
	\left(n-n/\log n \right)^{-1}\left\|\left(J_n^\tau\right)^{1/2}DT_k^2KD
	\Lambda^{-1/2}\right\|_F^2 \notag \\
	& \quad \quad \leq
	\left(n-n/\log n \right)^{-1}
	\Tr\left(\left(J_n^\tau\right)^{1/2}DT_k^2KT_k^2D\left(J_n^\tau\right)^{1/2}\right) \notag \\
	& \quad \quad \leq
	\left(n-n/\log n \right)^{-1}
	\Tr\left(T_k^4KDJ_n^\tau D\right)
	+\frac 12 \left(n-n/\log n \right)^{-1}
	\Tr\left(D\bar T_k DJ_n^\tau\right) \notag  \\
	& \quad \quad \leq 
	\left(n-n/\log n \right)^{-2}
	\Tr\left(T_k^4\right)
	+2\left(n-n/\log n \right)^{-1}\max_{i,j=1,\ldots,n-1}
	\left|\bar T_k\right|_{i,j}\Tr\left(J_n^\tau\right),
\end{align*}
where we used in the last inequality an argument as for (\ref{eq.t2}). Combining this with (\ref{eq.lbtau}) and Proposition \ref{prop.prop21} yields
\begin{align}
	&\sup_{\tau^2\in \Theta_s^b\left(\beta, \bar Q\right)} \max_{k\leq n^{1/2}}
	\left|
	\left\|\left(J_n^\tau\right)^{1/2}D\left(T_k^2K+KT_k^2\right)D\left(J_n^\tau\right)^{1/2}\right\|_F^2
	- 
	\frac 2n \int_0^1\tau_k^4(x)dx\right| \notag \\
	&\quad\quad \quad \quad  =O\left(n^{-1}\bar \phi_{n,1/2}+n^{-1}\log n \bar \phi_{n,1/2}^2
	+n^{-1}\log^{-1} n\right).
	\label{eq.maintermo}
\end{align}
Now we will bound the remainder term in (\ref{eq.normineqs}). Using Lemma \ref{lem.eigbound} gives
\begin{eqnarray*}
	\left\|\left(J_n^\tau\right)^{1/2}DS_{k}D\left(J_n^\tau\right)^{1/2}\right\|_F^2
	&\leq& \lambda_1^2\left(J_n^\tau\right)\left\|S_{k}\right\|_F^2 \\
	&\leq&
	16n^{-2}\log^4 n \left(\left\|\tilde T_{k}\right\|_F^2
	+4\left\|\Cov\left(Z_{2,k}\right)\right\|_F^2
	\right. \\
	&& \left. \quad\quad+8\left\|\Cov\left(X_{2,k},Z_{2,k}\right)\right\|_F^2\right).
\end{eqnarray*}
Because $\Cov\left(X_{2,k},Z_{2,k}\right)$ is tridiagonal it holds with Lemma \ref{lem.ABineq} (i)
\begin{align*}
	\left\|\Cov\left(X_{2,k},Z_{2,k}
		\right)\right\|_F^2
	&=\sum_{i,j=1}^{n-1}\left(\Cov\left(X_{2,k},Z_{2,k}
	\right)_{i,j}\right)^2 
	\leq 8n\tau_{\max}\phi_{n,1/2}^2
\end{align*}
and therefore
\begin{eqnarray*}
	&&\left\|\left(J_n^\tau\right)^{1/2}DS_{k}D\left(J_n^\tau\right)^{1/2}\right\|_F^2 \\
	&&\quad\quad \leq 16n^{-2}\log^4 n
	\left(2n \bar \phi_{n,1/2}^4
	+16n \bar \phi_{n,1/2}^4
	+64n \bar \phi_{n,1/2}^2
	\tau_{\max}\right)
\end{eqnarray*}
This leads to
\begin{align*}
	&\sup_{\tau^2\in \Theta_s^b\left(\beta, \bar Q\right)} \max_{k\leq n^{1/2}}  \left\|\left(J_n^\tau\right)^{1/2}DS_{k}D\left(J_n^\tau\right)^{1/2}\right\|_F^2  = O\left(n^{-1}\log^4 n \bar \phi_{n,1/2}^2\right)
\end{align*}
and with (\ref{eq.normineqs}) and (\ref{eq.maintermo}) completes the proof.
\end{proof}

\begin{lem2}
\label{lem.approx}
Let $s_{k,p}$ and $t_{k,p}$ as defined in (\ref{def.skl}). Then it holds
\begin{itemize}
\item[(i)]
\begin{align*}
	s_{k,p}= \frac 12 s_{0,p}+\frac 14 s_{0,p-k} +\frac 14 s_{0,p+k}, \quad
	t_{k,p}= \frac 12 t_{0,p}+\frac 14 t_{0,p-k} +\frac 14 t_{0,p+k}.
\end{align*}
\item[(ii)]
\begin{align*}
	\frac 1n \sum_{r=1}^{n-1} \sigma^2_k\left(\frac rn\right)
	\cos\left(\frac{pr\pi}n\right)
	&= A\left(\sigma_k^2, p\right) -\frac 1{2n}
	\left(\left(-1\right)^p \sigma^2_k(1)+\sigma^2_k(0)\right).
\end{align*}
\item[(iii)] Let $\Sigma_k$ as defined in (\ref{eq.sigmak}). Then 
\begin{align*}
	\left(D \Sigma_k^2 D\right)_{i,j}
	= A\left(\sigma_k^2, i-j\right)-A\left(\sigma_k^2, i+j\right).
\end{align*}
\end{itemize}
\end{lem2}

\begin{rem2}
In $(iii)$, for $\left|i-j\right| \ll i+j$, the r.h.s. behaves like $s_{k,i-j}$. In the same way we obtain the equivalent result if we replace $\sigma^2$ by $\tau^2$.
\end{rem2}

\begin{proof}
(ii) Note that we can write
\begin{eqnarray*}
	\sigma_k^2\left(\frac rn\right)
	=s_{k,0}+2\sum_{q=1}^\infty s_{k,q}\cos\left(q\pi r/n\right)
\end{eqnarray*}
and hence it holds
\begin{align*}
	&\frac 1n \sum_{r=1}^{n-1} \sigma_k^2\left(\frac rn\right) 
	\cos\left(\frac {p r \pi}n\right) \\
	&\quad =
	\frac 1n s_{k,0}\sum_{r=1}^{n-1} \cos\left(\frac {p r \pi}n\right)
	+\frac 2n \sum_{q=1}^\infty s_{k,q} 
	\sum_{r=1}^{n-1} \cos\left(\frac{q\pi r}n\right)
	\cos\left(\frac {p r \pi}n\right).
\end{align*}
Let $\mathbb{I}_{\left\{A\right\}}$ denote the indicator function on the set $A$. We have the identities
\begin{eqnarray*}
	\sum_{r=1}^{n-1}\cos\left(\frac{pr\pi}n\right)
	= n\mathbb{I}_{\left\{p \equiv 0 \bmod 2n\right\}}-\frac 12\left(1+\left(-1\right)^p\right)
\end{eqnarray*}
and
\begin{eqnarray*}
	2 \sum_{r=1}^{n-1} \cos\left(\frac{q\pi r}n\right)
	\cos\left(\frac {p r \pi}n\right)
	= \sum_{r=1}^{n-1} \cos\left(\frac{\left(q-p\right)\pi r}n\right)
	+ \sum_{r=1}^{n-1} \cos\left(\frac{\left(q+p\right)\pi r}n\right).
\end{eqnarray*}
From this it follows
\begin{align*}
	&\frac 1n \sum_{r=1}^{n-1} \sigma^2\left(\frac rn\right) 
	\cos\left(\frac {p r \pi}n\right) \\
	&\quad\quad =
	\frac 1n \left[-\frac 12\left(1+\left(-1\right)^p\right) s_{k,0}
	-\sum_{q=1}^\infty s_{k,q}\left(1+\left(-1\right)^{q-p}\right)
	\right]
	+A\left(\sigma_k^2,p\right),
\end{align*}
which yields the result.

(iii) This follows by applying $(ii)$ to
\begin{align*}
	\left(D
	\Sigma_k^2 D\right)_{i,j}
	&=
	\frac 2n\sum_{r=1}^{n-1} \sigma^2_k\left(\frac rn\right)
	\sin\left(\frac{ir\pi}n\right)
	\sin\left(\frac{rj\pi}n\right) \notag \\
	&= 
	\frac 1n \sum_{r=1}^{n-1} \sigma^2_k\left(\frac rn\right)
	\cos\left(\frac{\left(i-j\right)r\pi}n\right)
	-
	\frac 1n \sum_{r=1}^{n-1} \sigma^2_k\left(\frac rn\right)
	\cos\left(\frac{\left(i+j\right)r\pi}n\right).
\end{align*}
\end{proof}

The next Lemma gives a bound of the absolute values of Fourier coefficients of $\sigma_k^2$ in Sobolev s-ellipsoids. In particular the result shows that the Fourier series is absolute summable.

\begin{lem2}
\label{lem.bound}
Let $s_{k,p}$ be as defined in (\ref{def.skl}). Assume $k \leq cn^{\gamma}$, where $0<c<1$ is a constant and either $\gamma>0, \alpha>1/2$ or $k=0, \gamma=0$ and $\alpha>1/2$. Then it holds for $n$ large enough
\begin{eqnarray*}
	\sup_{{\sigma^2\in \Theta_s\left(\alpha,Q\right)}}
	\sum_{m=\left[n^{\gamma}\right]}^\infty \left|s_{k,m}\right|
	\leq C_{\gamma, \alpha, Q, c}n^{\gamma\left(1/2-\alpha\right)},
\end{eqnarray*}
where $C_{\gamma, \alpha, Q, c}$ is independent of $n$.
\end{lem2}

\begin{proof}
Consider the case $\gamma>0, \alpha>1/2$. Using Lemma \ref{lem.approx} (i), we see that for $n$ large enough
\begin{eqnarray*}
	\sum_{m=\left[n^{\gamma}\right]}^\infty \left|s_{k,m}\right|
	&\leq& \sum_{m=\left[\left(1-c\right)n^{\gamma}\right]}^\infty
	\left|s_{0,m}\right|
	= \sum_{m=1}^\infty \left|s_{0,m}\right|
	I_{\left\{m \geq \left[\left(1-c\right)n^{\gamma}\right]\right\}} \\
	&\leq& 2
	\left(\sum_{i=1}^\infty i^{2\alpha} \frac{s_{0,i}^2}4\right)^{1/2}
	\left(\sum_{i=\left[\left(1-c\right)n^{\gamma}\right]}^\infty
	i^{-2\alpha}\right)^{1/2} 
	\leq C_{\gamma, \alpha, Q, c}n^{\gamma\left(1/2-\alpha\right)},
\end{eqnarray*}
where we used the definition of a Sobolev s-ellipsoid in the last step. If $k=0$, $\gamma=0$ and $\alpha>1/2$ we can argue similarly.

\end{proof}

In the next lemma we collect some important facts about positive semidefinite matrices and trace calculation.
\begin{lem2}
\label{lem.p}
\begin{itemize}
\item[(i)]
Let $A \in \mathbb{M}_n$ be symmetric. $A$ is positive semidefinite iff $A=B^tB$ for some $B \in \mathbb{M}_n$.
\item[(ii)]
If $A, B$ are positive semidefinite matrices. Denote by $\lambda_1(A)$ the largest eigenvalue of $A$. Then $\Tr(AB)\leq \lambda_1(A)\Tr(B)$.
\item[(iii)] Let $A, B \in \mathbb{M}_{n-1}$ be positive semidefinite. Then 
\begin{eqnarray*}
	\lambda_{r+s+1}\left(AB\right)
	&\leq& \lambda_{r+1}\left(A\right)\lambda_{s+1}\left(B\right)
	\quad 0\leq r+s\leq n-2 \\
	\lambda_{n-r-s+1}\left(AB\right)
	&\geq& \lambda_{n-r}\left(A\right)\lambda_{n-s}\left(B\right)
	\quad 2\leq r+s\leq n.
\end{eqnarray*}
\item[(iv)] Let $A$ and $B$ symmetric matrices. Then
\begin{eqnarray*}
	\lambda_{r+s+1}\left(A+B\right)
	\leq \lambda_{r+1}\left(A\right)+\lambda_{s+1}\left(B\right)
	\quad 0\leq r+s\leq n-2 .
\end{eqnarray*}
\item[(v)] (CS inequality for trace operator) Let $A$ and $B$ matrices of the same size. Then
\begin{eqnarray*}
	\left|\Tr\left(AB^t\right)\right|
	\leq \Tr^{1/2}\left(AA^t\right)\Tr^{1/2}\left(BB^t\right).
\end{eqnarray*}
\item[(vi)] Let $A, B$ matrices of the same size. Then $$A^tB+B^tA\leq A^tA+B^tB.$$ 
\end{itemize}
\end{lem2}

\begin{cor2}
\label{cor.eig}
Let $A$ and $B$ matrices of the same size. Then
\begin{eqnarray*}
	\lambda_1\left(AB^t+BA^t\right)
	\leq \lambda_1\left(AA^t\right)+\lambda_1\left(BB^t\right).
\end{eqnarray*}
\end{cor2}
\begin{proof}
By Lemma \ref{lem.p} $(vi)$  $A^tB+AB^t\leq A^tA+B^tB$. Applying Lemma \ref{lem.p} $(iv)$ for $r=s=0$ yields the result.
\end{proof}

In the following Lemma, we summarize some facts on Frobenius norms.
\begin{lem2}
\label{lem.ABineq}
\label{lem.froblem}
\begin{enumerate} Let $A \in \mathbb{M}_{n-1}$. Then
\item[(i)]
\begin{eqnarray*}
	\left\|A\right\|_F^2:=\Tr\left(AA^t\right)
	=\sum_{i=1}^{n-1} \lambda_i\left(AA^t\right)
	=\sum_{i,j=1}^{n-1} a^2_{i,j}
\end{eqnarray*}
and whenever $A=A^t$ also $\left\|A\right\|_F^2=\sum_{i=1}^{n-1} \lambda_i^2\left(A\right)$.
\item[(ii)]
It holds
\begin{eqnarray*}
	4\Tr\left(A^2\right)\leq \left\|A+A^t\right\|_F^2\leq 4\left\|A\right\|_F^2.
\end{eqnarray*}
\item[(iii)]
Let $A$, $B$ be positive semidefinite matrices of the same size and $0\leq A\leq B$. Further let $X$ be another matrix of the same size. Then
\begin{eqnarray*}
	\left\|X^tAX\right\|_F
	\leq
	\left\|X^tBX\right\|_F.
\end{eqnarray*}
\end{enumerate}
\end{lem2}

\begin{proof}
(i) and (ii) is well known and omitted. (iii) By assumptions it holds $0\leq X^t A X \leq X^t B X$. Hence $\lambda_i^2\left(X^t A X\right)\leq \lambda_i^2 \left( X^t B X\right)$ and the result follows.
\end{proof}

%
%

\begin{lem2}
\label{lem.n}
Let $V=\left(V_1,\ldots,V_n\right), W=\left(W_1,\ldots,W_m\right)$ be two independent, centered random vectors. Let $A=\left(a_{i,j}\right)_{i,j=1,\ldots,n}\in \mathbb{M}_{n}$ and $B \in \mathbb{M}_{n,m}$. Then
\begin{itemize}
\item[(i)] $\E\left(V^tAV\right)=\Tr\left(A\Cov\left(V\right)\right)$, $\E\left(V^tBW\right)=0$ and
\item[(ii)] Assume further that $V_i \perp V_j$ for all $i,j=1,\ldots,n$, $i\neq j$ and $W_k \perp W_l$ for all $k,l=1,\ldots,m$, $k\neq l$ and $\Var\left(V_i\right)=\Var\left(W_k\right)=1$ for $i=1,\ldots,n$ and $j=1,\ldots,m$. We set $\TrSq(A):=\sum_{i=1}^{n}a_{i,i}^2.$ Then  
\begin{align}
	\Var\left(V^tAV\right)
	& =  \Cum_4\left(V\right)\sum_{i=1}^n a_{ii}^2+\Tr\left(A^2+AA^t\right)
	\notag \\
	&\leq \Cum_4\left(V\right)\sum_{i=1}^n a_{ii}^2 + 2 \left\|A\right\|_F^2 \notag \\
	&\leq \left(2+\Cum_4\left(V\right)\right)\left\|A\right\|_F^2, 
	\label{eq.vtav} \\ 
	\Var\left(V^tBW\right)&=\left\|B\right\|_F^2, \notag \\
	\Var\left(V^tABW\right)
	&\leq \left\|AA^t\right\|_F \left\|BB^t\right\|_F.
	\label{eq.vtabw}
\end{align}
\end{itemize}
\end{lem2}

\begin{proof}
We only proof the first and the last statement in $(ii)$. Note that $$\Var\left(V^tAV\right)= \sum_{i,j,k,l=1}^n a_{ij}a_{kl}\Cov\left(V_iV_j,V_kV_l\right).$$ If $i=j=k=l$ then $\Cov\left(V_iV_j,V_kV_l\right)= 2+\Cum_4\left(V\right)$; if $i=k$, $j=l$, $i\neq j$ or $i=l$, $j=k$, $i\neq j$ then $\Cov\left(V_iV_j,V_kV_l\right)=1$. Otherwise $\Cov\left(V_iV_j,V_kV_l\right)=0$ and this gives (\ref{eq.vtav}).

In order to see (\ref{eq.vtabw}) note that by Lemma \ref{lem.p} (v)
\begin{align*}
 \Var\left(V^tABW\right)
	&=\left\|B^tA\right\|_F^2
	= \Tr\left(\left(BB^t\right)\left(AA^t\right)\right) \\
	&\leq \Tr^{1/2}\left(\left(BB^t\right)^2\right)
	\Tr^{1/2}\left(\left(AA^t\right)^2\right)
	= \left\|BB^t\right\|_F\left\|AA^t\right\|_F.
\end{align*}
\end{proof}
%
%
%

\begin{thm2}
\label{thm.wclt}
Let $\xi \sim \mathcal{N}(0,I_n)$ and $A$ be a positive semidefinite matrix. Then 
\begin{eqnarray*}
	\Var^{-1/2}\left(\xi^t A \xi\right)
	\left(\xi^t A \xi- \E\xi^t A \xi \right)
	\rightarrow
	\mathcal{N}(0,1)
\end{eqnarray*}
if and only if $\Var^{-1/2}\left(\xi^t A \xi\right)\lambda_1\left(A\right) \rightarrow 0$.
\end{thm2}

\begin{lem2}
\label{lem.eigbound}
Let $n\geq 4$. Then
\begin{eqnarray*}
	\lambda_1\left(J_n^{\tau}\right)\leq 4n^{-1}\log^2n.
\end{eqnarray*}
\end{lem2}

\begin{proof}
Let $r=\left[n/\log n\right]$ and note that $\sin(x)^{-1}\leq 2/x$ for $x\in (0,\pi/2]$. Then
\begin{align*}
	\lambda_r^{-1}
	\leq \left(\frac{2n}{r\pi}\right)^2
	\leq \frac 4{\pi^2}n^2\left(\frac 12 \frac n{\log n}\right)^{-2}
	\leq 2\log^2 n
\end{align*}
and
\begin{align*}
	\lambda_1\left(J_n^{\tau}\right)
	&=
	\left(n-n/\log n \right)^{-1}\lambda_r^{-1}
	\leq \frac 4n \log^2 n	
\end{align*}
\end{proof}

\begin{lem2}
\label{lem.sumeig}
Let $\lambda_i$ be as defined in (\ref{def.lambdai}). Then it holds
\begin{eqnarray*}
	\sqrt n \sum_{i=\left[n^{1/2}\right]+1}^{2\left[n^{1/2}\right]}
	\lambda_i
	=
	\frac{7\pi^2}3+O\left(n^{-1/2}\right).
\end{eqnarray*}
\end{lem2}

\begin{proof}
Let $x_i=i\pi/\left(2n\right)$. Note that $\sin^2\left(x_i\right)= x_i^2 -\xi_i^4/3$, where $\xi_i \in \left(0,x_i\right)$. Further $\max_{i=\left[n^{1/2}\right]+1,\ldots,2\left[n^{1/2}\right]} x_i \leq n^{-1/2}\pi$. Hence
\begin{eqnarray*}
	n^{1/2}\sum_{i=\left[n^{1/2}\right]+1}^{2\left[n^{1/2}\right]}\xi_i^4 
	\leq n \max_{i=\left[n^{1/2}\right]+1,\ldots,2\left[n^{1/2}\right]} x_i^4=O\left(n^{-1}\right)
\end{eqnarray*}
and thus
\begin{eqnarray*}
	n^{1/2}\sum_{i=\left[n^{1/2}\right]+1}^{2\left[n^{1/2}\right]} \lambda_i
	&=& 4n^{1/2} \sum_{i=\left[n^{1/2}\right]+1}^{2\left[n^{1/2}\right]}
	\frac{i^2\pi^2}{4n^2} +\frac 13\xi_i^4 \\
	&=&\pi^2n^{-3/2} \sum_{i=\left[n^{1/2}\right]+1}^{2\left[n^{1/2}\right]} i^2 +O\left(n^{-1}\right)
	= \frac{7\pi^2}3+O\left(n^{-1/2}\right).
\end{eqnarray*}

\end{proof}

%

\begin{lem2}[Continuous Sobolev Embedding]
\label{lem.sobolem}
Let $\mathcal{C}\left(q\right)$, $q>0$ denote the space of H\"older continuous functions on $[0,1]$ equipped with the canonical norm $\left\|.\right\|_{\mathcal{C}\left(q\right)}$ and define
\begin{align*}
	\eta: (1/2,\infty)\times [0,\infty) 
	\rightarrow \mathbb{R}, \quad
	\eta\left(\alpha,\delta\right):=
	\begin{cases}
	\alpha-1/2 \quad &\alpha \in \left(1/2,3/2\right), \\
	1-\delta &\alpha=3/2, \\
	1 &\alpha>3/2.
	\end{cases}
\end{align*}
Suppose $\alpha>1/2$. Then for any $\delta>0$ the embedding 
\begin{align*}
	\iota : \ \Theta_s^b\left(\alpha,Q\right)
	\hookrightarrow \mathcal{C}\left(\eta\left(\alpha,\delta\right)\right)
\end{align*}
is continuous and in particular
\begin{align*}
	\sup_{f\in \Theta_s^b\left(\alpha,Q\right)}
	\left\|f\right\|_{\mathcal{C}\left(\eta\left(\alpha,\delta\right)\right)}
	<\infty.
\end{align*}
\end{lem2}

\begin{proof}
For a given function $f:\left[0,1\right]\rightarrow \mathbb{R}$ define $\tilde f:\left[-1,1\right]\rightarrow \mathbb{R}$,
\begin{align*}
	\tilde f(x):=
	\begin{cases}
	f(x) \quad &x\in\left[0,1\right],\\
	f(-x) &x\in\left[-1,0\right].
	\end{cases}
\end{align*}
Let for $s>0,$ $\left.W^{s,2}\left[-1,1\right]\right|_{\left[0,1\right]}$ denote the (fractional) Sobolev space on $\left[-1,1\right],$ where the domain of functions is restricted to $\left[0,1\right]$ equipped with the norm
\begin{align*}
	\left\|f\right\|_{W^{s,2}\left[-1,1\right]\mid_{\left[0,1\right]}}
	:= \left\| \tilde f \right\|_{W^{s,2}\left[-1,1\right]}.
\end{align*}
Note this is a function space on $\left[0,1\right]$ and $\left.W^{s,2}\left[-1,1\right]\right|_{\left[0,1\right]} \neq W^{s,2}\left[0,1\right]$. By the Sobolev embedding theorem (see Taylor (1996), Proposition 8.5) we have for $\alpha>1/2$ that
\begin{align*}
	\iota : \ \Theta_s^b\left(\alpha,Q\right) \subseteq \left. W^{\alpha,2}\left[-1,1\right]\right |_{\left[0,1\right]}
	\hookrightarrow
	\mathcal{C}\left(\eta\left(\alpha,\delta\right)\right)
\end{align*}
is continuous and since it is linear also bounded. This yields
\begin{align*}
	\sup_{f\in \Theta_s^b\left(\alpha,Q\right)}
	\left\|f\right\|_{\mathcal{C}\left(\eta\left(\alpha,\delta\right)\right)}
	\leq \left\| \iota \right\| \sup_{f\in \Theta_s^b\left(\alpha,Q\right)}
	\left\|f\right\|_{W^{\alpha,2}\left[-1,1\right]\mid_{\left[0,1\right]}}
	<\infty.
\end{align*}
\end{proof}

\end{appendices}

\end{document}